\documentclass{ecai}
\usepackage{graphicx}
\usepackage{latexsym}

\usepackage{textcomp}
\usepackage{times}
\usepackage{soul}
\usepackage{url}
\usepackage[hidelinks]{hyperref}
\usepackage[utf8]{inputenc}
\usepackage[small]{caption}
\usepackage{graphicx}
\usepackage{amsmath}
\usepackage{amsthm}
\usepackage{booktabs}
\usepackage{algorithm}
\usepackage[switch]{lineno}
\usepackage{mathrsfs}
\usepackage{proofzilla}
\newtheorem{definition}{Definition}
\newtheorem{example}{Example}
\newtheorem{proposition}{Proposition}
\newtheorem{thm}{Theorem}

\newtheorem{lemma}{Lemma}
\usepackage[utf8]{inputenc}
\usepackage[noend]{algpseudocode}

\let\defi\emph 
\let\abbrv\textsf 

\newcommand*{\crossupsidedown}{%
    \textbf{%
      \raisebox{.5ex}{%
        \makebox[0pt][l]{%
          \rule[-.2pt]{.75ex}{.4pt}%
        }%
        \makebox[.75ex]{%
          \rule[-.5ex]{.4pt}{1.5ex}%
        }%
      }%
    }%
}

\def\tradatl#1{(#1)^\circ}
\def\Ap{\mathsf{AP}}
\def\Ag{\mathsf{Ag}}
\def\sem#1{ \llbracket #1 \rrbracket }
\def\sabotage{\crossupsidedown}
\def\tuple#1{\langle #1 \rangle}
\def\set#1{\{#1\}}
\def\next{\mathsf{X}\,}
\def\until{\, \mathsf{U} \,}
\def\untilJ{\, \mathsf{U_J} \,}
\def\release{\, \mathsf{R}\,}
\def\releaseJ{\, \mathsf{R_J}\,}
\def\globaly{\mathsf{G}\,}

\def\eventualy{\mathsf{F}\,}

\def\M{\mathfrak{M}}
\def\N{\mathfrak{N}}

\def\strat{\mathfrak{S}}

\def\path{\mathcal{P}}
\def\strat{\mathfrak{S}}
\def\naww#1{\langle   #1  \rangle }

\def\exis#1{\langle \! \langle #1 \rangle \! \rangle}

\def\tradatl#1{(#1)^{\mathsf{A}}}

\def\tradTmu#1{(#1)^{T\mu}}

\newemptyvertex{node}{draw,circle}
\defedgetype{grafo}{->,line width=0.35mm}{}
\defedgetype{grafob}{->,line width=0.35mm,blue,dotted}{}
\newemptyvertex{circle}{draw,circle}
\newemptyvertex{win}{draw,circle,fill, red}
\newemptyvertex{pholds}{draw,circle,fill,blue}

\newemptyvertex{goal}{draw,circle,blue}

\algnewcommand\algorithmicswitch{\textbf{switch}}
\algnewcommand\algorithmiccase{\textbf{case}}

\algdef{SE}[SWITCH]{Switch}{EndSwitch}[1]{\algorithmicswitch\ #1\ \algorithmicdo}{\algorithmicend\ \algorithmicswitch}%
\algdef{SE}[CASE]{Case}{EndCase}[1]{\algorithmiccase\ #1}{\algorithmicend\ \algorithmiccase}%
\algtext*{EndSwitch}%
\algtext*{EndCase}%

\begin{document}

\begin{frontmatter}
\onecolumn

\title{A Timed Obstruction Logic for Dynamic Game Models}

\author[A]{\fnms{David}~\snm{Cortes}\orcid{0009-0002-5771-8817}}
\author[A]{\fnms{Jean}~\snm{Leneutre}}
\author[A]{\fnms{Vadim}~\snm{Malvone}\orcid{0000-0001-6138-4229}} 
\author[A]{\fnms{James}~\snm{Ortiz}\thanks{Corresponding Author. Email: james.ortizvega@telecom-paris.fr}}
\address[A]{LTCI, Institut Polytechnique de Paris, Télécom Paris}
 
\begin{abstract}
%
Real-time cybersecurity and privacy applications require reliable verification methods and system design tools to ensure their correctness. Many of these reactive real-time applications embedded in various infrastructures, such as airports, hospitals, and oil pipelines, are potentially vulnerable to malicious cyber-attacks. Recently, a growing literature has recognized Timed Game Theory as a sound theoretical foundation for modeling strategic interactions between attackers and defenders. This paper proposes Timed Obstruction Logic (\abbrv{TOL}), an extension of Obstruction Logic (\abbrv{OL}), a formalism for verifying specific timed games with real-time objectives unfolding in dynamic models. These timed games involve players whose discrete and continuous actions can impact the underlying timed game model. We show that \abbrv{TOL} can be used to describe important timed properties of real-time cybersecurity games. Finally, in addition to introducing our new logic and adapting it to specify properties in the context of cybersecurity, we provide a verification procedure for \abbrv{TOL} and show that its complexity is \abbrv{PSPACE-complete}, meaning that it is not higher than that of classical timed temporal logics like \abbrv{TCTL}. Thus, we increase the expressiveness of properties without incurring any cost in terms of complexity.

\end{abstract}

\end{frontmatter}

\section{Introduction}

Multi-agent systems (\abbrv{MAS}) are usually understood as systems composed of interacting autonomous agents~\cite{KVW01, AHK02, LomuscioQuRaimondi09, MogaveroMPV14, Jamroga15modulestrats, Jennings+98b}. In this sense, \abbrv{MAS} have been successfully applied as a modeling paradigm in several scenarios, particularly in the area of cybersecurity and distributed systems. However, the process of modeling security and heterogeneous distributed systems is inherently error-prone: thus, computer scientists typically address the issue of verifying that a system actually behaves as intended, especially for complex systems.

Some techniques have been developed to accomplish this task: testing is the most common technique, but in many circumstances, a formal proof of correctness is required. Formal verification techniques include theorem checking and model checking~\cite{ClarkeGrumbergPeled99}. In particular, model checking techniques have been successfully applied to the formal verification of security and distributed systems, including hardware components, communication protocols, and security protocols. Unlike traditional distributed systems, formal verification techniques for Real-Time \abbrv{MAS} (\abbrv{RT-MAS}) \cite{CRS21, CTh17} are still in their infancy due to the more sophisticated nature of the agents, their autonomy, their real-time constraints, and the richness of the formalisms used to specify properties~\cite{CE81, KVW00}.

\abbrv{RT-MAS} combines game theory techniques with real-time behavior in a distributed environment \cite{CRS21, CTh17, GNC20}. Such real-time behavior can be verified using real-time modal logic \cite{OAS19, APS23}.  The development of methods and techniques with integrated features from both research areas undoubtedly leads to an increase in complexity and the need to adapt current techniques or, in some cases, to develop new formalisms, techniques, and tools \cite{QAS19, QAF15, ABZ07}. Developing these formalisms correctly requires algorithms, procedures, and tools to produce reliable end results \cite{StrategicTool, DynamicVadim}.

Agents in \abbrv{RT-MAS} are considered to be players in games played over real-time models (such as Timed Automata (\abbrv{TA}) \cite{alurD94} and Timed Petri Nets \cite{KnapikPetrucciJamrogaTimed}), and their goals are specified by real-time logic formulas \cite{NNA23, QAF15, ABZ07, KnapikPetrucciJamrogaTimed, APS23}. For example, the fact that a coalition of players has a strategy to achieve a certain goal by acting cooperatively can be expressed using the syntax of logics such as Timed Alternating-time Temporal Logic (\abbrv{TATL})~\cite{HPV06,LaroussinieMO06}. For example, the fact that a coalition of players has a strategy to achieve a certain goal by acting cooperatively can be expressed using the syntax of logics such as \abbrv{TATL} and Strategic Timed Computation Tree Logic (\abbrv{STCTL})~\cite{APS23}. However, \abbrv{STCTL} with continuous semantics is more expressive than \abbrv{TATL}, as shown in \cite{APS23}. Moreover, in \cite{APS23} was shown that the model checking problem for \abbrv{STCTL}  with continuous semantics and memoryless perfect information is of the same complexity as for \abbrv{TCTL}, while for \abbrv{STCTL} with continuous semantics and perfect recall is undecidable. Model checking for \abbrv{TATL} with continuous semantics is undecidable \cite{APS23}.  
In all previous logics, the timed game model in which the players are playing is treated as a static game model, i.e., the actions of the players affect their position within the model, but do not affect the structure of the model itself. 

In this paper, we propose a  new logic, Timed Obstruction Logic (\abbrv{TOL}), for reasoning about \abbrv{RT-MAS} with real-time goals \cite{NNA23, QAF15, ABZ07, KnapikPetrucciJamrogaTimed, APS23} played in a dynamic model. Dynamic game models \cite{NelloPlanning, DynamicVadim, CattaLM23} have been studied in a variety of contexts, including cybersecurity and planning. In our new logic (\abbrv{TOL}), games are played over an extension \abbrv{TA} (Weighted \abbrv{TA} (\abbrv{WTA}) \cite{ASG01}) by two players (Adversary and Demon). There is a \abbrv{cost} ($W(e)$) associated with each edge of the automaton. This means that, given a location $l$ of the automaton and a natural number $n$, the Demon deactivates an appropriate subset $T$ of the set of edges incident to $l$ such that the sum of the deactivation costs of the edges contained in $T$ is less than $n$. Then, the Adversary selects a location $l'$ such that $l$ is adjacent to $l'$ and the edge from $l$ to $l'$ does not belong to the set of edges selected by the Demon in the previous round. The edges deactivated in the previous round are restored, and a new round starts at the last node selected by the opponent. The Demon wins the timed game if the infinite sequence of nodes subsequently selected by the opponent satisfies a certain property $\varphi$ expressed by a timed temporal formula. Furthermore, in addition to the introduction of our new logic and its adaptation to the specification of properties in the context of real-time cybersecurity, we provide a verification procedure for \abbrv{TOL} and show that its complexity is \abbrv{PSPACE-complete}, i.e., not higher than that of classical timed temporal logics such as \abbrv{TCTL}. Thus, we increase the expressiveness of properties without incurring a cost in terms of complexity.




\textbf{Structure of the work.} The contribution is structured as follows. 
Theoretical background is presented in Section~\ref{sec:background}. In Section~\ref{sec:tol}, we present the syntax and the semantics of our new logic, called Timed Obstruction Logic (\abbrv{TOL}). In Section~\ref{sec:mc}, we show our model checking  algorithm and prove that the model checking problem for \abbrv{TOL} is \abbrv{PSPACE-Complete}.  In Section~\ref{sec:relationship}, we compare \abbrv{TOL} with other timed logic for strategic and temporal reasoning. In Section~\ref{sec:example}, we present our case study in the cybersecurity context. In Section~\ref{sec:relwork}, we compare our approach to related work. Finally, Section~\ref{sec:end}  concludes and presents possible future directions.

\section{Background}
\label{sec:background}


\paragraph{General Concepts.} 
Let $\mathbb{N}$ be the set of natural numbers, we refer to the set of natural numbers containing $0$ as $\mathbb{N}_{\ge 0}$, $\mathbb{R}_{\ge 0}$ the set of non-negative reals and $\mathbb{Z}$ the set of integers.  
Let $X$ and $Y$ be two sets and $|X|$ denotes its cardinality. The set operations of intersection, union, complementation, set difference, and Cartesian product are denoted $X$ $\cap$ $Y$, $X$ $\cup$ $Y$, $\overline{X}$, $X$ $\setminus$ $Y$, and $X$ $\times$ $Y$, respectively. Inclusion and strict inclusion are denoted $X$ $\subseteq$ $Y$ and $X$ $\subset$ $Y$, respectively. The empty set is denoted $\emptyset$. Let $\pi=x_1,\ldots,x_n$ be a finite sequence, $last(\pi)$ denotes the last element $x_n$ of $\pi$, we use $\sqsubseteq$ (resp. $\sqsubset$) to denote the prefix relation (resp. strict prefix relation).  Let $\Sigma$ be a finite alphabet of actions. The set of all finite words over $\Sigma$ will be denoted by $\Sigma^{*}$. A timed action over an alphabet $\Sigma$ is a finite sequence $\theta$ = $((\sigma_1, t_1), (\sigma_2, t_2) \cdot (\sigma_n, t_n))$ of actions $\sigma_i$ $\in$ $\Sigma$ that are paired with non-negative real numbers $t_i$ $\in$ $\mathbb{R}_{\ge 0}$ such that the sequence $t = t_1t_2 \cdots t_n$ of time-stamps is non-decreasing (i.e., $t_i \le t_{i+1}$ for all $1 \le i < n$). 
\paragraph{Attack Graphs and Moving Target Defense Mechanisms.} A malicious attack is defined as an attempt made by an attacker to gain unauthorized access to resources or to compromise the integrity of the system as related to system security controls. In this context, the Attack Graph (\abbrv{AG}) \cite{KK2016}\cite{CattaLM23} is one of the most popular attack models that has been created and is receiving a lot of attention recently. 
Leveraging an \abbrv{AG}, it is possible to model the interactions between an attacker, and a defender able to dynamically deploy \textit{Moving Target Defense (\abbrv{MTD})} mechanisms~\cite{Cho2020}. \abbrv{MTD} mechanisms, such as \textit{Address Space Layout Randomization (\abbrv{ALSR})}~\cite{app9142928}, are active defense mechanisms that use partial reconfiguration of the system to dynamically alter the attack surface and lower the attack’s chance of success. 
As a drawback, the activation  of a \abbrv{MTD} countermeasure has an impact on system performance: during the reconfiguration phase the services of the system are partially or not available. It is therefore critical to be able to select \abbrv{MTD} deployment strategies both minimizing the residual cybersecurity risks and the negative impact on the performance of the system.

\subsection{Weighted Transition Systems} Weighted Transition Systems (\abbrv{WTS}) are an extension of the standard notion of (Labeled) Transition Systems (\abbrv{LTS}) \cite{PLK81}, which has been used to introduce operational semantics for a variety of reactive systems. An \abbrv{AG} can be defined as a \abbrv{WTS}. 
\begin{definition} [Weighted Transition Systems (\abbrv{WTS})] Let $\Ap$ be a finite set of atomic propositions. A \abbrv{WTS} is a tuple $\mathcal{M} = (\textit{S}, \textit{s}_{0}, \Sigma, \textit{E}, \textit{W}, K, \textit{F})$ where:

\begin{itemize} 
	 \item $S$ is a finite set of states,
    \item $s_0\in S$ is an initial state,
    \item $\Sigma$ is a finite set of actions, 
     \item $E \subseteq S \times \Sigma \times S$ is a transition relation, 
    \item $W$: $E \to \mathbb{N}_{\ge 0}$ is a function that labels the elements of $E$, 
     \item $K$: $S \to 2^{\Ap}$ is a labeling function for the states,
    \item $F \subseteq S$ is a set of goal states.
\end{itemize}	
\end{definition}	
	
\noindent 
The transitions from state to state of a \abbrv{WTS} are noted in the following way: we write $\textit{s} \xrightarrow[w]{a} \textit{s}'$ whenever $a$ $\in$ $\Sigma$, $(s, a, s')$ $\in$ $E$ and $W(s, a, s')$ $=$ $w$ where $w$ $\in$ $\mathbb{N}_{\ge 0}$. 
A path of $\mathcal{M}$ can be defined as a finite (resp. infinite) sequence of moves:  $\varrho$ = $s_0 \xrightarrow[w_1]{\textit{a}_1} s_1 \xrightarrow[w_2]{\textit{a}_2} s_2 \ldots s_{n-2} \xrightarrow[w_{n-1}]{\textit{a}_{n-1}} s_{n-1}$, where $\forall 0 \le i \le n-1$,  $\forall j \ge 1,$ $w_{j}$ $\in$ $\mathbb{R}_{\ge 0}$ and $a_{j} \in \Sigma$. A path is \defi{initial} if it starts in $s_{0}$. Thus, an initial path describes one execution of the system. A (Weighted) trace from $\varrho$  is a sequence $((a_1, w_1), (a_2, w_2), \ldots, (a_n, w_n))$ of pairs $(a_i, w_i)$ $\in$ $\Sigma \times  \mathbb{R}_{\ge 0}$ for which there exists a path $\varrho$ from which $w_i$ = $W(s_i, a_i, s_{i+1})$.

\begin{example}
Figure~\ref{fig:exp2} gives an example of a \abbrv{WTA}. States of the \abbrv{WTA} are denoted as $s_i$, with $0\leq i  \leq 5$, attack actions and weighted as edge labels $a_j$, with $1\leq j  \leq 7$, $w_j$, with $1 \leq j  \leq 11$.
\begin{figure}[h]
 \centering
  \includegraphics[width=74mm, height=25mm]{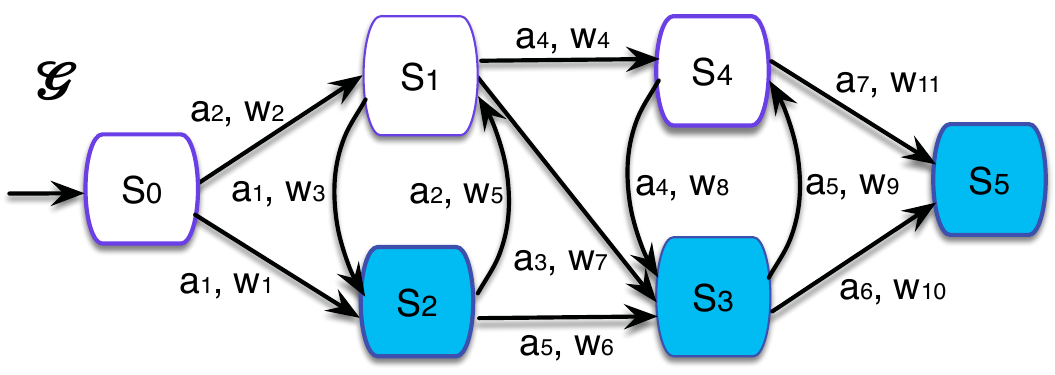}
    \vspace{-4pt}
  \caption{A \abbrv{WTS} where states $s_2,\, s_3$ and $s_5$ represents the goal states of the attacker.}
  \label{fig:exp2}
 \end{figure}
\end{example}



\paragraph{Clocks and Timed Automata.} We use non-negative real-valued variables known as \textit{clocks} to represent the continuous time domain. Clocks are variables that advance synchronously at a uniform rate; they are the basis of \abbrv{TA} \cite{alurD94}. Model checkers such as \abbrv{UPPAAL} \cite{behrmann2006tutorial}, \abbrv{KRONOS} \cite{bozga98}, and \abbrv{HYTECH} \cite{Henzinger97} support and extend \abbrv{TA}. 
Clock constraints within \abbrv{TA} control transitions. 
Here, we work with an extension of \abbrv{TA} known as Weighted Timed Automata (\abbrv{WTA}) \cite{ASG01}.

\subsection{Weighted Timed Automata}


We now explore the relation between \abbrv{WTS} and Weighted Timed Automata  (\abbrv{WTA}) \cite{ASG01}. A \abbrv{WTA} is an extension of a \abbrv{TA} \cite{alurD94} with weight/cost information at both locations and edges, and it can be used to address several interesting questions \cite{BFLM11, ASG01}.

\begin{definition}[Clock constraints and invariants]
Let $\textit{X}$ be a finite set of clock variables ranging over $\mathbb{R}_{\ge 0}$ (non-negative real numbers). Let $\Phi(\textit{X})$ be a set of clock constraints over  
$\textit{X}$. A \textit{clock constraint} $\phi\in\Phi(X)$ can be defined by the following grammar: 
\begin{center}
$\phi ::= true \ | \ \textit{x} \sim \textit{c} \ | \ \phi_{1} \ \wedge \ \phi_{2} $  \\
\end{center}
where  $\textit{x} \in \textit{X}$, $\textit{c} \in \mathbb{N}$, and $\sim  \in  \{<, >, \le, \ge, =\}$.

Clock invariants $\Delta(\textit{X})$ are clock constraints in which $\sim  \in  \{<,\le\}$.
\end{definition}



\begin{definition}[Clock valuations] Given a finite set of clocks $\textit{X}$, a clock valuation function, $\nu : \textit{X} \to \mathbb{R}_{\ge 0}$ assigning to each clock $x$ $\in$ $\textit{X}$ a non-negative value $\nu(x)$. We denote $\mathbb{R}_{\ge 0}^{X}$ the set of all valuations. For a clock valuation $\nu\in\mathbb{R}_{\ge 0}^{X}$ and a time value $\textit{d} \in  \mathbb{R}_{\ge 0}$, $\nu + \textit{d}$ is the valuation satisfied by $(\nu + d)(x) = \nu(x) + d$ for each $x$ $\in$ $\textit{X}$. Given a clock subset $Y \subseteq X$, we denote $\nu[Y \gets 0]$ the valuation defined as follows: $\nu[Y \gets 0](x) = 0$ if $x\in Y$ and  $\nu[\textit{Y} \gets 0](x) = \nu(x)$ otherwise.
\end{definition}

Here, we only consider the weight/cost in the edges (transitions) in our \abbrv{WTA}. Formally, a \abbrv{WTA} is defined as follows \cite{ASG01}.




\begin{definition}[Weighted Timed Automata (\abbrv{WTA})]
\label{def:wta}
Let $X$ be a finite set of clocks and $\Ap$ a finite set of atomic propositions.
A \abbrv{WTA} is a tuple $\mathcal{A}$ = $(L, l_0, X, \Sigma, T, I, W,  K,  F)$, where:
\begin{itemize}
    \item $L$ is a finite set of locations,
    \item $l_0\in L$ is an initial location,
    \item $X$ is a finite set of clocks,
    \item $\Sigma$ is a finite set of actions,
    \item $T \subseteq L  \times \Sigma \times \Phi (X) \times 2^{X} \times L$ is a finite set of edges (or transitions), 
    \item $I: L \rightarrow \Delta (X)$ is a function that associates to each location a clock invariant, 
     \item $W$: $T  \to \mathbb{N}_{\ge 0}$ is a function that labels the elements of $T$, 
      \item $K: L \to 2^{\Ap}$ is a labeling function for the locations,
    \item $F \subseteq L$ is a set of goal locations.
\end{itemize}
\end{definition}

We write $\textit{l} \xrightarrow[w]{a, \phi, Y} \textit{l}'$ instead of $(l, a, \phi, Y,l')_{w}$ $\in T$  for an edge from $l$ to $l'$ with guard $\phi$ $\in$ $\Phi(X)$, reset set $Y$ $\subseteq$ $X$ and $w$ $\in$ $\mathbb{N}_{\ge 0}$.  The value $W(t)$ given to edge $t$ = $(l, a, \phi, Y,l')_{w}$ where $t$ $\in$ $T$ represents the cost of taking that edge. In this paper, the value $W(t)$ given to edge $t$ represents the deactivation cost. 
Since cost information cannot be employed as constraints on edges, the undecidability of Hybrid Automata (\abbrv{HA}) \cite{Henzinger97} is avoided in the case of \abbrv{WTA} \cite{BFLM11} (i.e.,  decidability results are preserved for \abbrv{WTA}). In \abbrv{WTA}, costs are explicitly defined in its syntax, however, they do not influence the discrete behavior of the system. Since there is no cost constraint, the semantics of a \abbrv{WTA} is similar to that of a \abbrv{TA}. It is thus given as a \abbrv{WTS}. 

\begin{definition}[Semantics of \abbrv{WTA}]
\label{defSemTA}
Let $\mathcal{A} = (L, l_0, X, \Sigma, T, I, W, K, F)$ be a \abbrv{WTA}. The semantics of \abbrv{WTA} $\mathcal{A}$ is given by  a \abbrv{WTS}$({\mathcal{A}})$ = $(\textit{S}, \textit{s}_{0}, \Sigma_{\Delta}, E, W', K', \textit{S}_{F})$  where:

\begin{itemize}
    \item $S \subseteq L$ $\times$ $\mathbb{R}^{\textit{X}}_{\ge 0}$ is a set of states,
    \item $s_0= (l_0, \nu_0)$ with $\nu_0(x)=0 $ for all $x\in \textit{X}$ and $\nu_0$ $\models$ $I(l_0)$,
     \item $S_F \subseteq F$ $\times$ $\mathbb{R}^{\textit{X}}_{\ge 0}$ is a set of states,
    \item $\Sigma_{\Delta}$ = $\Sigma$   $\uplus$ $\mathbb{R}_{\ge 0}$,
    \item $W'$ = $W$, 
    \item $K'((l, \nu))$ = $K(l)$ $\cup$ $\{\phi \in \Phi(X) \mid \nu \models \phi \}$, 
    \item $E \subseteq S\times \Sigma_{\Delta} \times S$ is a transition relation defined by the following two rules:
    \begin{itemize}
        \item \textbf{Discrete transition: } $(l, \nu)\xrightarrow[w]{a} (l', \nu')$ 
        for $a\in \Sigma$ and $w$ $\in$ $\mathbb{N}_{\ge 0}$ iff $\textit{l} \xrightarrow[w]{a, \phi, Y} \textit{l}'$, $\nu$ $\models$ $\phi$, $\nu'$ $=$ $\nu[Y \gets 0]$ and $\nu'$ $\models$ $\textit{I}(l')$ and, 
        \item \textbf{Delay transition: } $(l, \nu)\xrightarrow{d} (l, \nu+d)$, for some  $d\in \mathbb{R}_{\geq 0}$ iff $\nu + d$ $\models$ $\textit{I}(l)$. 
    \end{itemize}
\end{itemize}

\end{definition}

\subsection{Paths and n-strategy}

A path $\rho$ in $\abbrv{WTS}(\mathcal{A})$ is an infinite sequence of consecutive delays and discrete transitions. A finite path fragment of $\mathcal{A}$ is a run in \abbrv{WTS}($\mathcal{A}$) starting from the initial state $s_0 = (l_0, \nu_0)$, with delay and discrete transitions alternating along the path:  $\rho$ = $s_0 \xrightarrow{d_0} s_1' \xrightarrow[w_0]{\textit{a}_0} s_1 \xrightarrow{d_1} s_2' \xrightarrow[w_1]{\textit{a}_1} s_2 \ldots s_{n-1} \xrightarrow{d_{n-1}} s_{n}'  \xrightarrow[w_n]{a_{n}} s_n$ $\ldots$ or more compactly  $s_0 \xrightarrow[w_0]{d_0, a_0} s_1 \xrightarrow[w_1]{d_1, \textit{a}_1} s_2 \xrightarrow[w_2]{d_2, a_2} s_3 \ldots s_{n-1} \xrightarrow[w_{n-1}]{d_{n-1}, a_{n-1}} s_{n}$ $\ldots$, where $\nu_0(x) = 0$ for every $x$ $\in$ $X$. A path of \abbrv{WTS}$({\mathcal{A}})$ is \textit{initial} if $s_0 = (l_0, \nu_0)$ $\in$ $S$, where $l_0$ $\in$ $L$, $\nu_0$ assigns 0 to each clock, and \textit{maximal} if it ends in a goal location. We write $\rho_i$ to denote the $i$-th element $s_i = (l_i, \nu_i)$  of $\rho$, $\rho_{\leq i}$ to denote the prefix $s_0,\ldots, s_{i}$ of $\rho$ and $\rho_{\geq i}$ to denote the suffix $s_i, s_{i+1}\ldots$ of $\rho$. 
A \emph{history} is any finite prefix of some path. We use $H$ to denote the set of histories.  

\begin{definition}
Let  $\mathcal{A}$ be a \abbrv{WTA} and $n$ be a natural number.  Given a model \abbrv{WTS}($\mathcal{A}$), a  $n$-strategy is a function $\strat: H\to 2^{E}$ that, given a history $h$, returns a subset $E'$ such that: 
    (i) $E' \subset E(last(h))$, 
    (ii) $(\sum_{e\in E'} \mathsf{W}(e)) \leq n$. 
 A memoryless $n$-strategy is a $n$-strategy $\strat$ such that for all histories $h$ and $h'$ if $last(h)=last(h')$ then $\strat(h)=\strat(h')$. 
\end{definition}

A path $\rho$ is compatible with a $n$-strategy  if for all $i$ $\ge 1$, $(\rho_i, \sigma, \rho_{i+1})$ $\notin$  $\strat(\rho_{\leq i})$, where $\sigma \in \Sigma_\Delta$.
Given a state $s = (l, \nu)$ and a $n$-strategy $\strat$, $Out(s,\strat)$ refers to the set of pathways whose first state is $s$ and are consistent with $\strat$.


\begin{example}
Let $\mathcal{M}$ be the \abbrv{WTA} depicted in Fig \ref{fig:wta}. $\mathcal{A}$ contains ten locations:  $L_0$ (initial) and $L_7$, $L_8$ and $L_9$ are the goal locations. For the sake, all locations define invariants to be true. The transition $L_0 \xrightarrow{a, (x \le 2), \{x:=0\}} L_1$ specifies that when the input action $a$ occurs and the guard $x \le 2$ holds, this enables the transition, leading to a new current location $L_1$, while resetting clock variables $x:=0$ and 3 as weight/cost. The transition $L_1 \xrightarrow{a, (x \le 2), \{x:=0\}} L_2$ specifies that when the input action $b$ occurs and the guard $x < 1$ holds, this enables the transition, leading to a new current location $l_2$, while resetting clock variables $x:=0$ and 2 as weight/cost. Likewise, the transition $L_1 \xrightarrow{c, (2< y < 3), \{y:=0\}} L_3$ specifies that when the input action $c$ occurs and the guard $2< y < 3$ holds, this enables the transition, leading to a new current location $L_3$, while resetting clock variables $y:=0$ and 1 as weight/cost. Likewise, we can see the other transitions in Fig \ref{fig:wta}.

\begin{figure}[h]
  \centering
  \includegraphics[width=82mm, height=47mm]{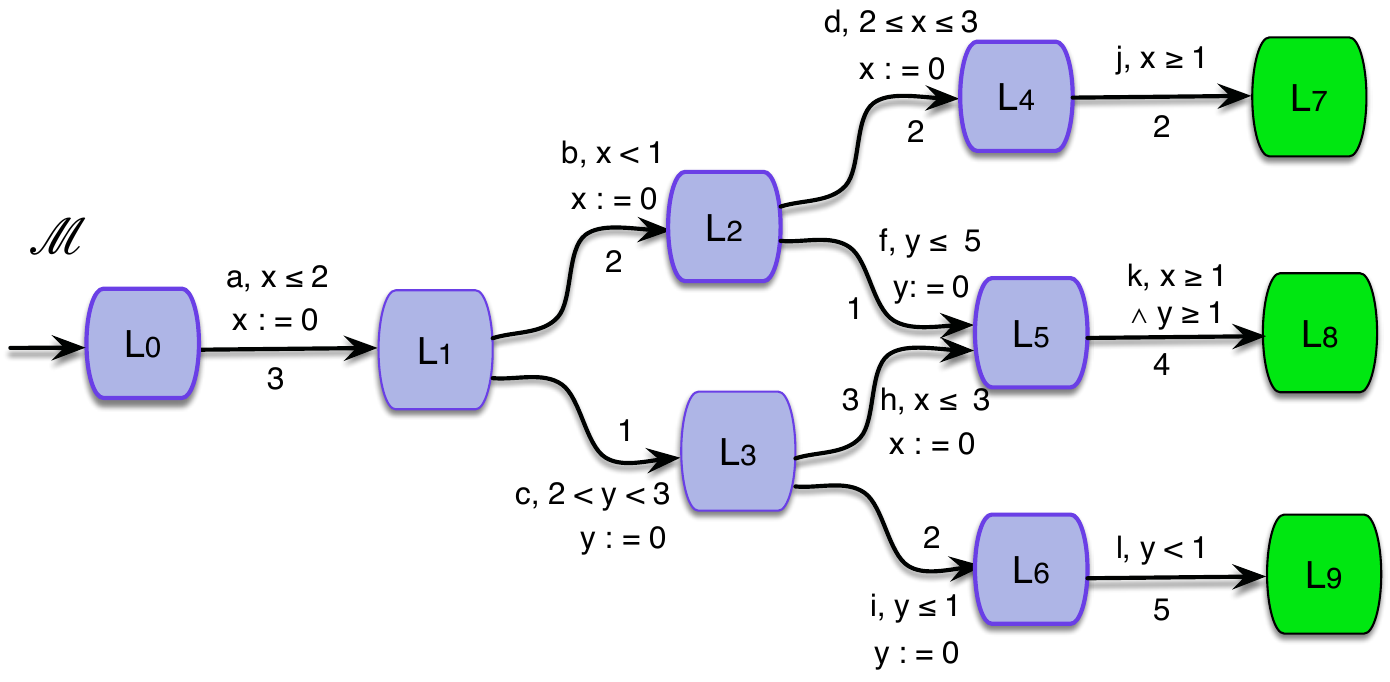}
    \vspace{-4pt}
  \caption{A \abbrv{WTA} with two clocks $x$ and $y$.}
  \label{fig:wta}
\end{figure}
\end{example}

\subsection{Predecessor and Zone Graph}
\label{subsection:graph}

Since the number of states in a \abbrv{WTA} is infinite, thus, it is impossible to build a finite state automaton. Then a symbolic semantics called \abbrv{zone graph} was proposed for a finite representation of \abbrv{TA} behaviors \cite{BYJ04}.  The \abbrv{zone graph} representation of \abbrv{TA} is not only an important implementation approach employed by most contemporary \abbrv{TA} tools \cite{BYJ04}, but it also provides a theoretical foundation for demonstrating the decidability of semantic properties for a given \abbrv{TA}. 

In a \abbrv{zone graph}, clock zones are used to symbolically represent sets of clock valuations. A clock zone $Z$ $\in$ $\mathbb{R}_{\ge 0}^{X}$ over a set of clocks $X$ is a set of valuations which satisfy a conjunction of constraints. Formally, the clock zone for the constraint $\phi$ is $Z$ = $\{ \nu  \ | \ \nu(x) \models \phi, x \in X \}$. Geometrically, a zone is a convex polyhedron. A symbolic state (or zone) is a pair $\mathcal{Z}$ = $(l, Z)$, where $l$ is a location and $Z$ is a clock zone. A zone $\mathcal{Z}$ = $(l, Z)$ represents all the states $z = (l', \nu)$ $\in$ $\mathcal{Z}$ if $l$ = $l'$ and $\nu$ $\in$ $Z$, indicating that a state is contained in a zone. Zones and their representation by \abbrv{Difference Bound Matrices} (\abbrv{DBM}s in short) are the standard symbolic data structure used in tools implementing real-time systems \cite{BYJ04}. We can now define the symbolic discrete  and delay predecessor operations on zones as follows:

\begin{definition}[Discrete and Time Predecessor] Let $\mathcal{Z}$ be a zone and $e$ be an edge of a \abbrv{WTS}($\mathcal{A}$), then: $$\abbrv{disc-pred}(e, \mathcal{Z}) = \{z \ | \   \exists z' \in \mathcal{Z}, z \xrightarrow[w]{e} z'\}$$
$$\abbrv{time-pred}(\mathcal{Z}) = \{z \ | \   \exists z' \in \mathcal{Z}, z \xrightarrow[w]{d} z' \ and \ d \in \mathbb{R}_{\ge 0} \}$$
\end{definition}

That is, $\abbrv{disc}$-$\abbrv{pred}(t, \mathcal{Z})$ is the set of all $e$-predecessors of states in $\mathcal{Z}$ and $\abbrv{time}$-$\abbrv{pred}(\mathcal{Z})$ is the set of all time-predecessors of states in $\mathcal{Z}$. According to these definitions, if $\mathcal{Z}$ is a zone then $\abbrv{time}$-$\abbrv{pred}(\mathcal{Z})$ and $\abbrv{disc}$-$\abbrv{pred}(e, \mathcal{Z})$ are also zones, meaning that zones are preserved by the above predecessor operations.




\begin{definition}[Predecessor] Let $\mathcal{Z}$ be a zone and $e$ be a edge of a \abbrv{WTS}($\mathcal{A}$), then : $$\abbrv{pred}(t, \mathcal{Z}) = \abbrv{disc-pred}(e, \abbrv{time-pred}(\mathcal{Z}))$$
\end{definition}

That is, $\abbrv{pred}$ is the set of all states that can reach some state in $\mathcal{Z}$ by performing a $e$ transition and allowing some time to pass.


\subsection{Obstruction Logic}\label{sec:logic}

In this section, we recall the syntax of obstruction logic. 

\begin{definition}
Let $\Ap$ be a set of atomic formulas (or atoms). Formulas of Obstruction Logic (\abbrv{OL}) are defined by the following grammar:
$$\varphi::=   \top \mid p \mid \neg \varphi\mid  \varphi \land \varphi  \mid \naww{\sabotage_n} \next \varphi \mid  \naww{\sabotage_n} (\varphi \until \varphi) \mid \naww{\sabotage_n}( \varphi \release  \varphi  )$$
where $p$ is an atomic formula and $n$ is any number in $\mathbb{N}_{\ge 0}$.
\end{definition}

The number $n$ is called \emph{the grade} of the strategic operator.  The boolean connectives $\bot$, $\vee$ and $\to$ can be defined as usual. The meaning of a formula $\naww{\sabotage} \varphi $ with $\varphi$ temporal formula is: there is a demonic strategy (that is, a strategy for disabling arcs) such that all paths of the graphs that are compatible with the strategy satisfy $\varphi$.  

\section{Timed Obstruction Logic}
\label{sec:tol}

In this section, we define the syntax and semantics of our Timed Obstruction Logic (\abbrv{TOL}). Our definitions are based on \cite{CLM23,CattaLeneutreMalvone,CattaLM23}.

\begin{definition}
Let $\mathcal{A}$ be a \abbrv{WTA}, $\Ap$ a set of atomic propositions (or atoms),  a set  $X$ of clocks of $\mathcal{A}$ and $J$ a non-empty set of clocks of the formula, where  $X$ $\cap$ $J$ = $\emptyset$. Formulas of Timed Obstruction Logic (\abbrv{TOL}) are defined by the following grammar:
$$\varphi::=   \top \mid p \mid \neg \varphi \mid  \varphi_1 \land \varphi_2  \mid  \phi \mid  \naww{\sabotage_n} (\varphi_1 \until \ \varphi_2) \mid \naww{\sabotage_n}( \varphi_1 \release \ \varphi_2) \mid j.\varphi$$ where $p$ $\in$ $\Ap$ is an atomic formula, $j$ $\in$ $J$, $n$ $\in$ $\mathbb{N}_{\ge 0}$ represents  \emph{the grade} of the strategic operator, and $\phi$ $\in$ $\Phi(X \cup J)$. 
\end{definition}


It is possible to compare a formula clock and an automata clock, for example, by using the clock constraint $\phi$, which applies to both formula clocks and clocks of the \abbrv{TA}.
The boolean connectives $\bot$, $\vee$ and, $\to$ can be defined as usual. Clock $j$ in $j.\varphi$ is called a freeze identifier and bounds the formula clock $j$ in $\varphi$. The interpretation is that $j.\varphi$ is valid in a state $s$ if $\varphi$ holds in $s$ where clock $j$ starts with value 0 in $s$. This freeze identifier can be used in conjunction with temporal constructs to indicate common timeliness requirements such as punctuality, bounded response, and so on. As \abbrv{OL}, we define $\naww{\sabotage_n}\eventualy \varphi:= \naww{\sabotage_n}(\top \until \varphi)$, $\naww{\sabotage_n} \globaly \varphi:= \naww{\sabotage_n} ( \bot \release  \varphi) $ and $\naww{\sabotage_n} (\varphi\, \mathsf{W} \, \psi):= \naww{\sabotage_n} (\psi \release (\varphi \vee \psi))$.  The size $|\varphi|$ of a formula $\varphi$ is the number of its connectives. With the help of the freeze identifier operator of $\abbrv{TOL}$, a time constraint can be added concisely. For instance, the formula $j.\naww{\sabotage_n}((\varphi_1 \wedge j \le 7) \until \varphi_2)$ intuitively means that there is a demonic strategy such that all paths that are compatible with the strategy, the property $\varphi_1$ holds continuously until within 7 time units $\varphi_2$ becomes valid. 
From the above formula, it is clear that timing constraints are allowed. In this case, we will call the formulas with timing constraints, such as \abbrv{timed temporal formulas}. The intuitive meaning of a formula $\naww{\sabotage} \varphi $ with $\varphi$ timed temporal formula is: there is a demonic strategy such that all paths of the \abbrv{WTS} that are compatible with the strategy satisfy $\varphi$.  Formulas of \abbrv{TOL} will be interpreted over \abbrv{WTS}. We can now precisely define the semantics of \abbrv{TOL} formulas.

\begin{definition}[\abbrv{TOL} Semantics]\label{def:sat}
Let $\mathcal{A}$ = $(L, l_0, X, \Sigma, T, I, W,  K,  F)$  be a \abbrv{WTA}, $p$ $\in$ $AP$, $\phi$ $\in$ $\Phi(X)$, $\mathcal{M}$ = \abbrv{WTS}($\mathcal{A}$). The satisfaction relation between a \abbrv{WTS} $\mathcal{M}$, a state $s$ = $(l, \nu)$ of $\mathcal{M}$, and \abbrv{TOL} formulas $\varphi$ and $\psi$ is given inductively as follows:
\begin{itemize}
    \item $\mathcal{M}, s\models \top$ for all state $s$,
    \item  $\mathcal{M}, s\models p$ iff $p\in K(s)$,
    \item  $\mathcal{M},s\models \neg \varphi$ iff not $\mathcal{M}, s\models \varphi$ (notation $\mathcal{M},s\not \models \varphi$), 
    \item  $\mathcal{M}, s \models \varphi_1$ $\wedge$ $ \varphi_2$ iff $\mathcal{M}, s\models \varphi_1$ and $\mathcal{M}, s\models \varphi_2$, 
    \item  $\mathcal{M}, s \models \phi$ iff  $\nu \models \phi$,
    \item  $\mathcal{M}, s\models \naww{\sabotage_n} (\varphi \until  \psi)$ iff there is a $n$-strategy $\strat$ such that for all $\rho\in Out(s,\strat)$ there is a $j\in \mathbb{N}$ such that $\mathcal{M}, \rho_j\models \psi $ and for all $0 \leq k < j$, $\mathcal{M},\rho_k\models \varphi$, 

    \item  $\mathcal{M}, s\models \naww{\sabotage_n} (\varphi \release \psi) $ iff there is a $n$-strategy $\strat$ such that for all $\rho\in Out(s,\strat)$ we have that either $\mathcal{M},\rho_i \models \psi$ for all $i\in \mathbb{N}$ or there is a $k\in \mathbb{N}$ such that $\mathcal{M},\rho_k \models \varphi$ and $\mathcal{M}, \rho_i \models \psi$ for all $0\leq i \leq k$, 

     \item  $\mathcal{M}, s \models j.\varphi$ iff  $\mathcal{M}, (l, \nu[j \gets 0]) \models \varphi$.
    \end{itemize}
    \end{definition}

    Two formulas $\varphi$ and $\psi$ are semantically equivalent (denoted by $\varphi \equiv \psi$) iff for any model $\mathcal{M}$ and state $s$ of $\mathcal{M}$,  $\mathcal{M}, s \models \varphi$ iff $\mathcal{M}, s \models \psi$.   Let $\varphi$ be any formula, $X$ a set of clocks and $\mathcal{A}$ be a \abbrv{WTA}, then   \abbrv{Sat}$({\varphi})$ denotes the set of states of $\mathcal{M}$ = $\abbrv{WTS}(\mathcal{A})$  verifying, $\varphi$, i.e., \abbrv{Sat}$(\varphi)=\set{s \in S \ \, | \, \ \mathcal{M}, s \models \varphi}$. 
    Next, we can establish the relationship between  \emph{local} model checking and \emph{global} model checking. 
 \begin{definition}
    Given $\mathcal{M}$, $s$, and $\varphi$, the \emph{local} model checking concerns determining whether $\mathcal{M}, s \models \varphi$. Given $\mathcal{M}$ and $\varphi$, the \emph{global} model checking concerns determining the set  \abbrv{Sat}$(\varphi)$.
\end{definition}

 \begin{definition}
	A state $s$ in a \abbrv{WTS}$(\mathcal{A})$ satisfies a formula $\varphi$ iff $(s, \mu_{0})$ $\models$ $\varphi$ where $\mu_{0}$ is the clock valuation that maps each formula clock to zero.
\end{definition}	

The relationship between \abbrv{WTA} and \abbrv{WTS} is defined as follows.

\begin{definition}
	Let $\mathcal{A}$ be a \abbrv{WTA} and  $\varphi$ $\in$ $\textsf{TOL}$, then $\mathcal{A}$ $\models$ $\varphi$ iff $\textsf{WTS}(\mathcal{A})$ $\models$ $\varphi$. 
\end{definition}

Let $\varphi$ be a formula, then the set of extended states satisfying $\varphi$ is independent of the valuation $\mu$ for the formula clocks. Thus, if $\varphi$ is a formula, then for any state $s$ = $(l, \nu)$ in a \abbrv{WTS} and valuations $\mu$, $\mu'$ for the formula clocks, we can get that  $\mathcal{M}, (l,\mu) \models$ $\varphi$ iff $\mathcal{M}, (l,\mu') \models$ $\varphi$. Therefore, it makes sense to talk about a state $s$ that satisfies $\varphi$. 

\section{Model Checking}\label{sec:mc}

Here, we present our model checking algorithm for \abbrv{TOL}. Furthermore, we show that the model checking problem for \abbrv{TOL} is decidable in \abbrv{PSPACE}-complete. The general structure of the algorithm shown here is similar to \abbrv{OL} model checking algorithm \cite{CLM23}.  \abbrv{TOL} model checking algorithm is based on the computation of the set \abbrv{Sat}$(\varphi)$ of all states satisfying a \abbrv{TOL} formula $\varphi$, followed by checking whether the initial state is included in this set.
A \abbrv{WTA} $\mathcal{A}$ satisfies \abbrv{TOL} state formula $\varphi$ if and only if $\varphi$ holds in the initial state of \abbrv{WTA}:
$\mathcal{A} \models \varphi$ if  and  only if $l_0 \in L$ such that, $(l_0, \nu_0)$ $\in$ \abbrv{Sat}($\varphi$), where $\nu_0(x) = 0$ for all $x \in X$. 
In our \abbrv{TOL}, a \abbrv{TOL} formula $\naww{\sabotage_n} ((\varphi \until \psi)$  holds in a state $s$ iff $\sabotage_n\abbrv{U}(\abbrv{Sat}(\psi_1), \abbrv{Sat}(\psi_2))$  with $\until$ being an \abbrv{TOL} operator. 
As mentioned in subsection~\ref{subsection:graph}, build a \abbrv{WTS}($\mathcal{A}$) for some \abbrv{WTA} $\mathcal{A}$ is therefore not practicable. Instead, the basic idea is to construct a zone graph \cite{BYJ04}, which is built from the \abbrv{WTA} $\mathcal{A}$ and the \abbrv{TOL} formula $\varphi$ (i.e., \abbrv{ZG}($\mathcal{A}$, $\varphi$)).
In short, Algorithm~\ref{alg:labeling} begins with a \abbrv{WTA} $\mathcal{A}$ and a formula $\varphi$ used to construct the zone graph \abbrv{ZG}($\mathcal{A}$, $\varphi$) and returns the set of states of $\mathcal{A}$ satisfying $\varphi$.  The Algorithm~\ref{alg:labeling} works as follows: it first constructs the zone graph \abbrv{ZG}($\mathcal{A}$, $\varphi$), then it recursively computes, for all subformulas $\psi$, the sets of states \abbrv{Sat}$(\psi)$ for which $\psi$ is satisfied. The computation of \abbrv{Sat}$(\psi)$ for $\psi$ being true, a proposition $p$, or a clock constraint $\phi$ is explicit. The negation and conjunction computations are straightforward.  The computation of the \abbrv{TOL} formula $\naww{\sabotage_n} (\psi_1 \until \psi_2)$ and $\naww{\sabotage_n}(\psi_1 \release \psi_2)$ are defined under the computation of predecessor sets. However, the notion of predecessors is different for the quantifiers in \abbrv{TCTL} \cite{alurD94}. The computation of the \abbrv{TOL} formula $\naww{\sabotage_n} (\psi_1 \until \psi_2)$  can be reduced to the computation of an \abbrv{OL} formula. The computation of $\naww{\sabotage_n} ((\psi_1 \until \ \psi_2)$  is a fixed-point iteration that starts at \abbrv{Sat}$( \psi_2)$ and iteratively adds all predecessor states that are in \abbrv{Sat}$(\psi_1)$.  We need to define a new predecessor operator to compute all predecessors with the obstruction operator. We will now use our  zone graph to compute predecessors.  The predecessor computation is done by the operator  $\blacktriangledown(n, \mathcal{Z})$ for a symbolic state  $\mathcal{Z}$ (zones) and a number $n$, computes the set of all predecessor states (likewise, for the operator $\release$). 

\begin{definition}
  Given a symbolic state  $\mathcal{Z}$ and $e$ an edge, we define $\abbrv{Pred}(\mathcal{Z})$ = $ \bigcup_{e \in  E}$ $ \abbrv{pred}(e, \mathcal{Z})$.   
\end{definition}

Therefore, the obstruction predecessor of a zone $\mathcal{Z}$, denoted $\blacktriangledown(n, \mathcal{Z})$, is defined as the set of symbolic states  that characterizes all predecessors of the symbolic state $\mathcal{Z}$, where each state $z$ satisfies $\blacktriangleright(z, n,\overline{\mathcal{Z}})$, that is, it can transition to a state not in the set $\mathcal{Z}$ where the sum of all successors of $s$ is less than or equal to $n$. 

\begin{definition}
Let $n \in \mathbb{N}$ and $\mathcal{Z}$ a symbolic state, we write:
$$ \blacktriangleright(z, n, \mathcal{Z}) \ = \ (\sum_{z' \ \in \ \mathcal{Z} \ \land \ \sigma \ \in \ \Sigma_\Delta} \mathsf{W}(z, \sigma, z'))\leq n$$   
$$\blacktriangledown(n,\mathcal{Z}) = \set{ z \in \abbrv{Pred}(\mathcal{Z}) \, \ | \,  \blacktriangleright(z, n,\overline{\mathcal{Z}})}$$ 
\end{definition}

\begin{proposition}
    Let  $z$ a state, $n$ a natural number, and $\mathcal{Z}$ a symbolic state (or zone), then $z$ $\in$ $\blacktriangledown(n,\mathcal{Z})$ iff $z$ $\in$ $\mathcal{Z}$.
\end{proposition} 
\begin{proof}(Sketch). A proof of this proposition may be obtained from the convexity of zones 
on \abbrv{TA} \cite{alurD94, TY01}. 
\end{proof}
  
\begin{proposition}
    Let  $z$ a state, $n$ a natural number, and $\mathcal{Z}$ a symbolic state, then if $\mathcal{Z}$ is convex, then $\blacktriangledown(n,\mathcal{Z})$ is also convex.
\end{proposition} 
\begin{proof}
      (Sketch). A proof of this proposition may be obtained from the convexity of zones 
on \abbrv{TA} \cite{alurD94, TY01}. 
\end{proof}
  
	\begin{algorithm}[ht]
		\caption{\abbrv{TOL} model checking  \\ \abbrv{Input:} A \abbrv{ZG}($\mathcal{A}$,  $\varphi$) where  $\mathcal{A}$ is a \abbrv{WTA} and  $\varphi$  is a \abbrv{TOL} formula\\
         \abbrv{Output:} \abbrv{Sat}$(\phi)$ $\gets$ $\{s \in S \ | \ s \models  \varphi \}$ }
		   \label{alg:labeling}
		    \begin{algorithmic}[1]
            \ForAll{ i $\le$  $|\varphi|$ }
             \ForAll{$\psi$  $\in Sub(\varphi) \ with \ |\psi| = i$}
			\Switch{$(\psi)$}
			\Case{$\psi=\top$}
			\State{$ \abbrv{Sat}(\psi) \gets S$}
			\EndCase
			\Case{$\psi=p$}
			\State{$  \abbrv{Sat}(\psi) \gets  \set{s\in S \, | \,  p\in K(s)}$}
			\EndCase
			\Case{$\psi=\neg \psi_1$}
			\State{$\abbrv{Sat}(\psi) \gets S \setminus \abbrv{Sat}(\psi) $}
			\EndCase
             \Case{$\psi= \phi$}
			\State{$\abbrv{Sat}(\psi) \gets \phi $}
			\EndCase
			\Case{$\psi=\psi_1 \land \psi_2$}
			\State{$\abbrv{Sat}(\psi) \gets \abbrv{Sat}(\psi_1)\cap \abbrv{Sat}(\psi_2)$}
			\EndCase
			\Case{$\psi=\naww{\sabotage_n} (\psi_1 \untilJ \psi_2)$}
			\State{$\abbrv{Sat}(\psi) \gets {\sabotage_n}\abbrv{U}(\abbrv{Sat}(\psi_1), \abbrv{Sat}(\psi_2))$}
			\EndCase
            \Case{$\psi=\naww{\sabotage_n}(\psi_1\releaseJ \psi_2)$}
			\State{$\abbrv{Sat}(\psi) \gets \sabotage_n\abbrv{R}(\abbrv{Sat}(\psi_1), \abbrv{Sat}(\psi_2))$}
			\EndCase
            \Case{$\psi=j.\psi_1$}
			\State{$\abbrv{Sat}(\psi) \gets \abbrv{Sat}(\psi_1)$}
			\EndCase 
			\EndSwitch
			\EndFor
          \EndFor
		\end{algorithmic}
	\end{algorithm}
For the $\until$ and $\release$ operators, auxiliary methods are defined. These methods are listed in Algorithm~\ref{alg:labeling1} and Algorithm~\ref{alg:labeling2}.  Algorithm~\ref{alg:labeling1} shows the backward search for computing the method ${\sabotage_n}\abbrv{U}(\abbrv{Sat}(\psi_1), \abbrv{Sat}(\psi_2))$  in line $15$ of Algorithm~\ref{alg:labeling}. 
\begin{algorithm}
		\caption{Backward search for computing  ${\sabotage_n}\abbrv{U}$  \\ \abbrv{Input:} A  \abbrv{TOL} formula $\naww{\sabotage_n} (\psi_1 \until \psi_2)$ \\
         \abbrv{Output:} \abbrv{Sat}$(\naww{\sabotage_n} (\psi_1 \until \psi_2))$ $\gets$ $\{s \in S \ | \ s \models  \naww{\sabotage_n} (\psi_1 \until \psi_2) \}$ }
		   \label{alg:labeling1}
		    \begin{algorithmic}[1]
			\State{$ X \gets \emptyset$}
            \State{$ Y \gets \psi_2$}
		    \While{$Y \neq X$}
   			\State{$X \gets Y$}
            \State{$Y \gets \psi_2 \cup (\psi_1 \cap \blacktriangledown (n,X))$ } 
			\EndWhile
		    \State{$ \mathbf{return} \  Y$}

		\end{algorithmic}
	\end{algorithm}
Algorithm~\ref{alg:labeling2} shows the backward search for computing the method ${\sabotage_n}\abbrv{R}(\abbrv{Sat}(\psi), \abbrv{Sat}(\psi_2))$  in line $17$ of Algorithm~\ref{alg:labeling}. 
\begin{algorithm}
		\caption{Backward search for computing   ${\sabotage_n}\abbrv{R}$  \\ \abbrv{Input:} A  formula $\naww{\sabotage_n} (\psi_1 \release \psi_2)$ \\
         \abbrv{Output:} \abbrv{Sat}$(\naww{\sabotage_n} (\psi_1 \release \psi_2))$ $\gets$ $\{s \in S \ | \ s \models  \naww{\sabotage_n} (\psi_1 \release \psi_2) \}$ }
		   \label{alg:labeling2}
		    \begin{algorithmic}[1]
			\State{$ X \gets \top$}
            \State{$ Y \gets \psi_2$}
		    \While{$Y \neq X$}
   			\State{$X \gets Y$}
            \State{$Y \gets \psi_2 \cap (\psi_1 \cup \blacktriangledown (n,X))$ } 
			\EndWhile
		    \State{$ \mathbf{return} \  Y$}

		\end{algorithmic}
	\end{algorithm}
Termination of the Algorithm~\ref{alg:labeling} intuitively follows, as the number of states in the zone graph is finite. The following proposition establishes the termination and the correctness of our model checking algorithm.

\begin{proposition}[Termination]
    Let $\mathcal{A}$ be a \abbrv{WTA} and $\varphi$ be a formula. Algorithm~\ref{alg:labeling} always terminates on input \abbrv{ZG}($\mathcal{A},\varphi$). 
\end{proposition}
\begin{proof} (Sketch).
  The Algorithm~\ref{alg:labeling} computes the zone graph $\mathcal{G}$ in a finite time. The computation of subformulas $\abbrv{Sub}(\psi)$ and the updating of the labeling function $K$ are also bounded, as the number of iterations is limited by finite sets. Since the  Algorithm~\ref{alg:labeling} terminates. 
 \end{proof}
\begin{proposition}[Soundness and Completeness]
Let $\mathcal{A}$ be a \abbrv{WTA}, $\varphi$ be a\abbrv{TOL} formula and  $\mathcal{M}$ = \abbrv{WTS}($\mathcal{A}$) be a \abbrv{WTS}. Assume the model checking (Algorithm~\ref{alg:labeling}) is sound and complete. Then, there exists  $s$ $\in$ $S$, such that, iff, $\mathcal{M}, s$ $\models$ $\varphi$. 
    \end{proposition}
 \begin{proof} (Sketch).
    In order to prove the theorem, it will be sufficient to show for every subformula $\psi$ of  $\varphi$ and every state  $s$ $\in$ $S$, iff, $\psi$ is true at $s$.
    \noindent ($\textbf{Soundness}$) For every $\psi$ $\in$ \abbrv{Sub}($\varphi$) and $s$ $\in$ $S$, implies $\mathcal{M}, s$ $\models$ $\psi$. We prove this by induction over the structure of $\psi$ as follows. The base cases, $\psi=\top$ and $\psi=p$ ($p$ $\in$ $\Ap$), are obvious. For the induction step, the cases of boolean combinations, $\psi= \neg \psi$ and $\psi=\psi_1 \wedge \psi_2$, of maximal state formulas is trivial. The induction step for the remaining obstruction operators is as follows.
     
     If $\psi$ = ${\sabotage_n}\abbrv{U}(\abbrv{Sat}(\psi_1), \abbrv{Sat}(\psi_2))$.  Let $Y$ be the set of states of $S$ that is returned by algorithm~\ref{alg:labeling1} at line $6$. We need to show that $Y=\abbrv{Sat}(\psi_2)$ provided that $X=\abbrv{Sat}({\psi_1})$. We first show that $\abbrv{Sat}({\psi}) \subseteq Y$. Suppose that $s\in \abbrv{Sat}(\psi)$. By the definition of satisfaction, this means that there is a strategy $\strat$ such that given any $\rho=s_1,s_2,\ldots $ in $Out(s,\strat)$. Note that since the cardinality of $\mathcal{M}$ is finite, and we can suppose that $\strat$ is memoryless, we can focus on the finite prefix $s_1,\ldots s_m$ of $\rho$ in which all the $s_i$ are distinct. Let $A_i$ ( for $i< |{\mathcal{M}}|)$ be the value of the variable $A$ before the first $i$-th iteration of the algorithm. We show that if $C\subseteq A_i$ then $C\subseteq A_{i+1}$. First of all, note that $A_i\subseteq \abbrv{Sat}(\psi_1)$ for all $i$.  By definition,  we have that $A_{i+1}=\blacktriangledown(n,A_i)\cap \abbrv{Sat}(\psi_1)$, i.e., $A_{i+1}$ is computed by taking all the element of $\abbrv{Sat}(\psi)$ that have at most $n$ successors that are not in $A_i$. Hence,  
      $\mathcal{M}, s\models \psi$. If $\psi$ = $\sabotage_n\abbrv{R}(\abbrv{Sat}(\psi_1), \abbrv{Sat}(\psi_2))$ then the proof is similar to the above case. 

      \noindent ($\textbf{Completeness}$) For every $\psi$ $\in$ \abbrv{Sub}($\varphi$) and $s$ $\in$ $S$, implies $\mathcal{M}, s$ $\not\models$ $\psi$ as follows. We prove this over the structure of $\psi$. The base cases, $\psi=\top$ and $\psi=p$ ($p$ $\in$ $\Ap$), are obvious. For the induction step, the cases of boolean combinations, $\psi= \neg \psi$, then $\psi$ was model checked and it was found to be true. Thus, $\mathcal{M}, s$ $\not\models$ $\psi$. For  $\psi=\psi_1 \wedge \psi_2$, then $\psi_1$ and $\psi_2$ were model checked and at least one of them was found to be false. Therefore, $\mathcal{M}, s$ $\not\models$ $\psi$.  The proof for $\psi$ = ${\sabotage_n}\abbrv{U}(\abbrv{Sat}(\psi_1), \abbrv{Sat}(\psi_2))$ then  $\mathcal{M}, s \not\models \psi$ is similar to the above case (similar for $\release$).

      
\end{proof}

The following theorem establishes the complexity of our model checking algorithm.

\begin{thm}\label{theoModelcheclTOL}
	The model checking problem of \abbrv{TOL} on \abbrv{WTA} is \abbrv{PSPACE}-complete.
\end{thm}

\begin{proof} \abbrv{PSPACE-hardness}:  The proof follows from the \abbrv{PSPACE-hardness} of the model-checking of the logic $\abbrv{TCTL}$ over \abbrv{TA} \cite{alurD94}, since \abbrv{WTA} \cite{ASG01} are an extension of \abbrv{TA} and \abbrv{TOL} is the corresponding extension of $\textsf{TCTL}$ and \abbrv{OL} \cite{CLM23}: If we take the 0-fragment of \abbrv{TOL} to be the set of \abbrv{TOL} formulas in which the grade of any strategic operator is $0$ (i.e., $\abbrv{TOL}^{0}$) then $\abbrv{TOL}^{0}$ = $\abbrv{TCTL}$ and \abbrv{WTA} = \abbrv{TA}.

 \abbrv{PSPACE-membership}: To prove \abbrv{PSPACE}-membership, we use the idea suggested in \cite{alurD94}. Let $\mathcal{A}$ be a \abbrv{WTA}, $\varphi$ $\in$ $\textsf{TOL}$, $D$ the number of clocks of the automaton $\mathcal{A}$, $C_x$ the maximal constant associated with of clocks $\mathcal{A}$ and $\varphi$, $m$ the nesting depth of the largest fixed-point quantifier in $\varphi$. We consider the zone graph $\abbrv{ZG}(\mathcal{A}, \varphi)$ \cite{BYJ04} associated with $\mathcal{A}$ and the formula $\varphi$ with clocks $X$. The zone graph depends on the maximum constants with which the clocks in $\mathcal{A}$ and $\varphi$ are compared. Using the zone graph $\abbrv{ZG}(\mathcal{A}, \varphi)$, model checking of $\abbrv{TOL}$ formulas can be done in polynomial time in the number of $D$, $C$, and $m$. This can be shown as in \cite{alurD94}.  According to \cite{APH02}, $\mathcal{A}$ $\models$ $\varphi$ iff $\mathcal{A}'$ $\models$ $\varphi$, where $\mathcal{A}'$ = $untimed(\mathcal{A})$ is the untimed automaton associated with $\mathcal{A}$ and $\varphi$ (the zone graph $\abbrv{ZG}(\mathcal{A}, \varphi)$ \cite{BYJ04}). The size of $\mathcal{A}'$ is polynomial in the length of the timing constraints of the given \abbrv{WTA} automaton and in the length of the formula $\varphi$ (assuming binary encoding of the constants), that is, $\vert \mathcal{A}' \vert$ = $O(\vert \varphi \vert \ \cdot \ (\vert L \vert + \vert T \vert) \ \cdot \ D! \ \cdot \ \prod_{x \in X} C_x)$. The zone graph $\mathcal{A}'$ can be constructed in linear time, which is also bounded by $O(\varphi \ \cdot \ (\vert L \vert + \vert T \vert) \ \cdot \ D! \ \cdot \prod_{x \in X} C_x)$ \cite{alurD94}. On the zone graph, untimed model checking can be done in time $O((\vert \varphi \vert \cdot \vert \mathcal{A}'\vert)$. Obviously, we get an algorithm of time complexity $O(\vert \varphi \vert \ \cdot \ (\vert L \vert + \vert T \vert))$. 
\end{proof}

\section{Relationship of TOL with other Logics}\label{sec:relationship}

In this section, we establish relative relation  between  \abbrv{TOL} with the  Timed Computation Tree Logic (\abbrv{TCTL}) \cite{alurD94}, Timed $\mu$-Calculus  ($\abbrv{T}_{\mu}$) \cite{HTN94} and Timed Alternating-Time Temporal Logic (\abbrv{TATL}) \cite{HPV06}.


\subsection{TOL and TCTL}

Here, we show that \abbrv{TOL} extends \abbrv{TCTL}\cite{alurD94} with a reduction to a fragment of our logic.   We define the $0$-fragment of \abbrv{TOL} to be the set of \abbrv{TOL}  formulae in which the grade of any strategic operator is $0$. We denote by \abbrv{TOL}$^0$ such a fragment. Let $(-)^{\bullet}$ be the mapping from  \abbrv{TOL}$^0$ to \abbrv{TCTL} formulas that translate each strategic operator $\naww{\sabotage_0}$ with the universal path operator $\mathsf{A}$ of \abbrv{TCTL}, i.e., the function recursively defined as follows.
\begin{definition}[Translation of \abbrv{TOL} to \abbrv{TCTL}]
Let $\varphi$ be a \abbrv{TOL} formula. Then \abbrv{TCTL} fragment formula $(\varphi)^{\bullet}$
is defined inductively as follows, where $p$ $\in$  $\Ap$ 

$$\begin{array}{l@{=\quad }l@{\qquad} l@{=\quad }l } 

    (\top)^\bullet& \top  \\

   (p)^\bullet  & p  \\
   
    (\neg \varphi)^\bullet & \neg (\varphi)^\bullet \\

    (\varphi_1 \land \varphi_2)^\bullet &(\varphi_1)^\bullet \land (\varphi_2)^\bullet \\

    (\phi)^\bullet &\phi \\ 
   

    (\naww{\sabotage_0} (\varphi_1\until \varphi_2))^\bullet &\mathsf{A} ((\varphi_1)^\bullet \until (\varphi_2)^\bullet)\\
    
    (\naww{\sabotage_0} (\varphi_1 \release  \varphi_2))^\bullet &
    \mathsf{A} ((\varphi_1)^\bullet \release (\varphi_2)^\bullet) \\
    (j.\varphi)^\bullet &j.(\varphi)^\bullet 
    
\end{array}$$
\end{definition}


\noindent Note that the function $(-)^\bullet$ induces a bijection between \abbrv{TOL} and \abbrv{TCTL} formulae. 
\begin{thm}
Let $\mathcal{A}$ be a \abbrv{WTA}. For every model $\mathcal{M}$ = \abbrv{WTS}$(\mathcal{A})$, state $s$, and formula $\varphi\in$ \abbrv{TOL}$^0$, we have that $\mathcal{M},s\models \varphi$ if and only if $\mathcal{M},s\models_{TCTL} (\varphi)^\bullet$, where $\models_{TCTL}$ is the \abbrv{TCTL} satisfaction relation. 
\end{thm}
\begin{proof}
   The result follows by observing that, for any state $s$, the set of paths compatible with a $0$-strategy $\strat$ starting at $s$ is equal to the set of paths starting at $s$ and that given any two  $0$-strategies $\strat_1$ and $\strat_2$ we have that $Out(s,\strat_1)=Out(s,\strat_2)$. 
\end{proof}

\subsection{TOL and $T_{\mu}$}

$\abbrv{T}_{\mu}$ is an extension of the modal $\mu$-calculus~\cite{KupfermanVardiGraded} with clocks. The formulas $\abbrv{T}_{\mu}$ are built from state predicates by boolean connectives, a temporal next operator, the reset quantifier for clocks, and the least fix-point operator ($\mu$). Let $\Ap$ be a non-empty at most countable set of atomic propositions, and $\mathcal{V}$ a non-empty at most countable set of formula variables.  The formal definition of the formulas  is as follows.
$$\varphi ::= \top \mid  p \mid Y \mid \phi\mid \neg \varphi \mid \varphi_1 \wedge \varphi_2 \mid \ \varphi_1 \vartriangleright \varphi_2 \mid  j.\varphi \mid \mu\,Y. \varphi$$

\noindent where $p\in \Ap$, $Y\in \mathcal{V}$, $\phi$ $\in$ $\Phi(X)$ and $j$ $\in$ $X_{\varphi}$, where $X$ is the set of clocks of the automaton and $X_{\varphi}$ is a set of clocks of the formula. In $\mu\, Y. \varphi$ it is required that the variable $Y$ occurs in the scope of an even number of negations in $\varphi$.  An occurrence of a propositional variable that is within the scope of $\mu$ is said to be bound. Free variables are variables that are not bound. A $\abbrv{T}_{\mu}$ formula is closed if all its variables are bound.  Let $\varphi$ and $\psi$ two formulas, and suppose that no variable that is bound in $\varphi$ is free in $\psi$\footnote{One can always respect this constraint by renaming the bound variables of $\varphi$.}. We write $\varphi[\psi/Y]$ to denote the result of the substitution of $\psi$ to each free occurrence of $Y$ in $\varphi$. 
The greatest fix-point operator can be defined by $\nu \, Y.\varphi$ = 
The $\vartriangleright$ operator can be considered a (timed) next operator, where a state satisfies $ \varphi_1 \vartriangleright \varphi_2$ if one of its time successors has an action transition whose destination state satisfies $\varphi_2$, and every intermediate time successor (including this one) fulfills $\varphi_1$ or $\varphi_2$. 
Semantically, $\abbrv{T}_{\mu}$ formulas are interpreted in relation to states in the \abbrv{TTS} $\mathcal{M}$ associated with the \abbrv{TA} $\mathcal{A}$ ($\mathcal{M}$ = \abbrv{WTS}($\mathcal{A}$)). An assignment $\alpha: \mathcal{V}\to 2^S $,  is a function that sends propositional variables to subsets of $S$. Given an assignment $\alpha$,  a subset $U$ of $S$, and a variable $R$, $\alpha[Y \to U]$ is the assignment defined by $\alpha[Y\to U](R)=U$ if $R=X$ and $\alpha[X\to U](R)=\alpha(R)$ otherwise. The satisfaction relation between a Timed Transition Systems (\abbrv{TTS}) (a class of submodels of \abbrv{WTS}) $\mathcal{M}$ \cite{HMA92}, a state $s$ = $(l, \nu)$ of $\mathcal{M}$ and $\abbrv{T}_{\mu}$ formula $\varphi$ ($\mathcal{M}, s \models_{\alpha}$ $\varphi$) is given inductively as follows:
\begin{itemize}
    \item $\mathcal{M}, s\models_{\alpha} \top$ for all state $s$,
    \item  $\mathcal{M}, s\models_{\alpha} p$ iff $p\in K(s)$,
    \item  $\mathcal{M}, s\models_{\alpha} Y$ iff $\alpha(Y)$,
    \item  $\mathcal{M}, s \models_{\alpha} \phi$ iff  $\nu \models \phi$,
    \item  $\mathcal{M},s\models_{\alpha}  \neg \varphi$ iff not $\mathcal{M}, s\models_{\alpha}  \varphi$ (notation $\mathcal{M},s\not \models_{\alpha}  \varphi$), 
    \item  $\mathcal{M}, s \models_{\alpha}  \varphi_1$ $\wedge$ $ \varphi_2$ iff $\mathcal{M}, s\models_{\alpha}  \varphi_1$ and $\mathcal{M}, s\models_{\alpha}  \varphi_2$, 
     \item  $\mathcal{M}, s \models_{\alpha}  \varphi_1 \vartriangleright \varphi_2$ iff  for some states $s'$, $s''$ $\in$ $S$, some delay $\delta$ $\in$ $\mathbb{R}_{\ge 0}$ and some $b$ $\in$ $\Sigma$ such that $s \xrightarrow{\delta} s' \xrightarrow{\textit{b}} s''$  and $\mathcal{M}, s'' \models_{\alpha}  \varphi_2$ and for all $\delta'$ $\in$ $\mathbb{R}_{\ge 0}$, $0$ $\le$ $\delta'$ $\le$ $\delta$ then $\mathcal{M}, s' \models_{\alpha}  \varphi_1 \vee \varphi_2$   
     \item  $\mathcal{M}, s \models_{\alpha}  j.\varphi$ iff  $\mathcal{M}, (l, \nu[j \gets 0]) \models_{\alpha} \varphi$,
    \item  $\mathcal{M}, s \models_{\alpha}  \mu\,Y. \varphi$ iff $ s \in \bigcap\set{Q \subseteq S \mid  \{s' =(l', \nu') \in S  \mid \mathcal{M}, (l', \nu'[Y \gets Q])} \models_{\alpha} \varphi  \} \subseteq Q \}$
    \end{itemize}
\def\K{\mathfrak K}
Here, we prove that \abbrv{TOL} extends $\abbrv{T}_{\mu}$ with a reduction to a fragment of our logic. More precisely, we show how to translate each \abbrv{TOL} formula $\varphi$ to a $\abbrv{T}_{\mu}$ formula $\tradTmu{\varphi}$ and that given a \abbrv{WTA} $\mathcal{A}$ such that $\mathcal{M}$ = \abbrv{WTS}($\mathcal{A}$). we have that $\abbrv{Sat}{\varphi} =\abbrv{Sat}{\tradTmu{\varphi}}$. The formal syntax of $\abbrv{T}_{\mu}$  is the following. 
Now,  let $\tradTmu{-}$ be the function from  \abbrv{TOL} formulas to $T_{\mu}$ formulas,  defined as follows:
$$\begin{array}{l@{=\quad }l@{\qquad} l@{=\quad }l }
   \tradTmu{\top}& \top  \\
    \tradTmu{p} & p  \\
      \tradTmu{\phi} & \phi \\
    \tradTmu{\neg \varphi} & \neg \tradTmu{\varphi} \\
    \tradTmu{\varphi_1 \land \varphi_2} &\tradTmu{\varphi_1} \land \tradTmu{\varphi_2} \\
     \tradTmu{j.\varphi} & j.\tradTmu{\varphi} \\
    \tradTmu{\naww{\sabotage_0} (\varphi_1\until \varphi_2)}& \mu\, Y. (\tradTmu{\varphi_2} \vee (\tradTmu{\varphi_1} \vartriangleright Y))\\
    \tradTmu{\naww{\sabotage_0}( \varphi_1  \release \varphi_2))} & \nu\, Y (\tradTmu{\varphi_2} \land ( \tradTmu{\varphi_1} \vartriangleright Y))   \\ 
\end{array}$$

\noindent Note that $\tradTmu{\varphi}$ is a closed  $T_{\mu}$ formula for every formula $\varphi$, thus the function $\tradTmu{-}$ has as image a proper fragment of  $T_{\mu}$. Let us call \emph{unary} a \abbrv{TOL} model $\mathcal{M}$  such that $W(t)=1$ for all $t\in T$. We have the following 


\begin{thm}
    If $\mathcal{M}$ is a unary \abbrv{WTS} then for every \abbrv{TOL} formula $\varphi$ we have that $\mathcal{M}, s \models \varphi$  iff  $\mathcal{M}, s \models \tradTmu{\varphi}$. 
\end{thm}




\subsection{TOL and TATL}

Here, we compare  our \abbrv{TOL} with \abbrv{TATL} \cite{HPV06}. In particular, we show that given a \abbrv{TOL} formula $\varphi$ and a \abbrv{WTA} $\mathcal{A}$ ($\mathcal{M}$ = \abbrv{WTS}($\mathcal{A}$)) that satisfies it, there is a Concurrent Game Structure (\abbrv{CGS})\cite{HPV06} that satisfies a \abbrv{TATL} translation of $\varphi$. First, define a rooted \abbrv{TOL} as a pair $\tuple{\mathcal{M}, s}$ where $\mathcal{M}$ is a \abbrv{WTS} and $s$ is one of its states. Given a natural number $n$, let $S^{\leq n}$ be the subset of $S \times 2^E$ defined by $(s,E)\in S^{\leq n}$ iff either $E=\emptyset$ or each $e\in E$ has $s$ as source  and $(\sum_{e\in E} W(e)) \leq n$.   If $\tuple{\mathcal{M},s}$ is a rooted \abbrv{TOL} model and $n$ is a natural number, then $\mathbb{G}^n_\mathcal{M}=\tuple{Q, q_i, \Ap,\Ag, act_D, act_T, P, \delta, \mathcal{V} }$ is the \abbrv{CGS}, where:

\begin{itemize}
    \item $Q=Q_D \cup Q_T$ is a set of states, where $Q_D=S$ and $Q_T= S^{\leq n} $. Moreover, $q_I=s$ is the initial state. The set $Q_D$ is the set of states where is the Demon's turn  to move, while $Q_T$ is the set of states in which is the Traveler's turn to move,
    \item $\Ap$ is a set of atomic formulas labeling states of $\mathcal{M}$, 
     \item $\Ag=\set{D,T}$ where $D$ is the Demon and $T$ is the Traveler, 
    \item The set of actions $act_D$ of the Demon is equal to the set of subset of $R$ appearing in $S^{\leq n}$ plus the idle action $\star$. More precisely $act_D=\set{E\in 2^R \, : \, \exists q \in S^{\leq n} \land q=\tuple{s,E} } \cup \set{\star}$,
    \item The set of actions $act_T$ of the Traveler is $R\cup \set{\star}$. We denote by $act = act_D \cup act_T$,
    \item The protocol function $P: Q \times \Ag \to 2^{act}\setminus \emptyset $ is defined as follows. For every $q\in Q_D$, we have that $P(q,i)$  is equal to  $X_q=\set{ E\in 2^R \, : \, \tuple{q,E}\in S^{\leq n}}$ if $i=D$, and $\set{\star}$ otherwise. For every $q\in Q_T$, we have that if $q=\tuple{s,E}$ then $P(q,i)$ is equal to $\set{ e\in R \, : \, e\not\in E \land s'\in S  }$ if $i=T$ and it is equal to $\set{\star}$ otherwise,  
    \item The transition function $\delta: Q \times act_D \times act_T \to Q$ is defined as follows: $\delta(q,E,\star)=\tuple{q,E}$ iff $q\in Q_D$ and $\delta(q,\tuple{s,s'},\star)= s'$ iff $q=\tuple{s,E}\in Q_T$ and $\tuple{s,s'}\notin E$; 
    \item The labeling function $V: S\to 2^{\Ap}$ is defined by $V(q)= K(q)$ for any $q\in Q_D$ and $V(q)=\emptyset$ for any $q\in Q_T$. 
\end{itemize}
Remark that given a  $\mathcal{M}$ and a natural number $n$, the \abbrv{CGS} $\mathbb{G}^n_\mathcal{M}$ can have a number of states that is \textbf{exponential} in the number of states of $\mathcal{M}$.  Consider the  function from \abbrv{TOL} formulas to \abbrv{TATL} formulas, inductively defined by:

$$\begin{array}{l@{=\quad }l@{\qquad} l@{=\quad }l }
   \tradatl{\top}& \top  \\
    \tradatl{p} & p  \\
     \tradatl{\phi} & \phi \\
    \tradatl{\neg \varphi} & \neg \tradatl{\varphi} \\
    \tradatl{\varphi_1 \land \varphi_2} &\tradatl{\varphi_1} \land \tradatl{\varphi_2} \\

      \tradatl{j.\varphi} & j.\tradatl{\varphi} \\
    \tradatl{\naww{\sabotage_n} (\varphi_1\until \varphi_2)}& \exis{D} {{\tradatl{\varphi}} \until {\tradatl{\psi}}}\\

    \tradatl{\naww{\sabotage_n}( \varphi_1 
 \release \varphi_2))} &
    \exis{D} {\tradatl{\varphi}}\release{\tradatl{\psi}} 
    
\end{array}$$

Given a \abbrv{CGS} $\mathbb{G}^n_\mathcal{M}$ as the one defined above, and a path $\rho$ of the \abbrv{CGS}, we write $\rho^D$ for the subsequence of $\rho$ containing only states that are in $Q_D$. If $\Delta$ is a \abbrv{TATL} strategy and $q\in Q_D$ is a state, then $Out^D (q,\Delta)$ denotes the set of sequences $\set{\rho \in Q_D^\omega \, : \, \rho= \pi^D \text{ for some } \pi\in Out(q, \Delta) }$. For a \abbrv{TATL} formula $\psi$,  we write $\mathbb{G}^n_\mathcal{M}, q\models_D \psi$ iff either: 

\begin{enumerate}
    \item $\psi$ is a boolean formula and $\mathbb{G}^n_\mathcal{M},q \models_{TATL} \psi$, where $\models_{TATL}$ is the standard ATL satisfiability relation or
    \item $\psi$ is a strategic formula $\exis{D} \psi_1$ and there is a strategy $\Delta$ such that for all $\rho\in Out^D(q,\Delta)$ we have that $\rho$ satisfies $\psi_1$ (where the specific clauses for the temporal connectives  $\release$ and $\until$ can be easily obtained).
\end{enumerate}

  We can now prove the following the lemma.

\begin{lemma}
    Let $\varphi$ be any \abbrv{TOL} formula that contains at most a strategic operator $\naww{\sabotage_n}$, we have that $\mathcal{M},s \models \varphi$ iff $ \mathbb{G}^n_\mathcal{M} ,s \models_{D} \tradatl{\varphi}$. 
     
\end{lemma}

\section{Case study}\label{sec:example}
Based on an \abbrv{AG} as presented in Section~\ref{sec:background}, we would like to check whether there are \abbrv{MTD} response strategies to satisfy some security objectives. To achieve this, we assume that:
$(1)$ The defender always knows the \abbrv{AG} state reached by the attacker (called \emph{attacker current state}). $(2)$ At every moment, there is a unique attacker current state in the \abbrv{AG}.
$(3)$ When detecting the attacker current state, the defender can activate a (or a subset of) \abbrv{MTD}(s) temporarily removing an (a subset of) outgoing edge(s). The defender cannot remove edges that are not outgoing from the attacker current state. $(4)$ The sum of the costs associated to the subset of \abbrv{MTD}s activated is less than a given threshold. $(5)$ When the attacker launches an attack from its current state, if the corresponding edge has not been removed by the defender, then the attack always succeeds (i.e. the attacker reaches the next state). $(6)$ When the attacker launches an attack from its current state, if the corresponding edge has been removed by the defender, then the attack always fails (i.e. the attacker stays in its current state).

Consider the model in Figure~\ref{fig:exp}. We can assume that when reaching state $s_1$, $s_3$, or $s_5$ the attacker has root privilege on a given critical server $s$. In addition, if the attacker completes attack steps $a_6$ or $a_7$ (that is, it reaches state $s_5$), then the defender will obtain information on the identity of the attacker. 
Let $a$ be an atomic proposition that expresses the fact that the identity of the attacker is known. Let $r_s$ be an atomic proposition expressing the fact that the attacker has root privilege on the server $s$.  
We can express, via \abbrv{TOL} formulas, the following security objectives: 

 \begin{itemize}
     \item \emph{The attacker will never be able to obtain root privileges on server s unless the defender can obtain information about his identity within $3$ time units}: that is, either we want the attacker to never reach a state satisfying $r_s$ \emph{or} if the attacker reaches such a state, the defender wants to be able to identify it within 3 time units ($a$). By using $t_1$ as a variable, the following OL formula captures the objective: $\varphi_1:= j.\naww{\sabotage_{t_1}}\globaly( \neg r_s \vee (r_s\to \naww{\sabotage_{t_1}} \mathsf{F ({ j \le 3}} \wedge a)))$. 

 \item \emph{While the defender has not obtained information about the attacker identity within $5$ time units, the attacker has not root privilege on server $s$}: that is, we want $r_s$ to be false \emph{until} we have identified the attacker ($a$) within 5 time units, if such an identification ever happens. Thus, by using $t_2$ as a variable for a given threshold, we can write our objective by using the weak-until connective: $\varphi_2:= j.\naww{\sabotage_{t_2}} (\neg r_s\, \wedge j \le 5 \  \mathsf{W}\, a)$.
 \end{itemize}
  
Suppose that $t_1$ and $t_2$ are respectively $3$ and $4$. Let $\mathcal{M}$ = $\abbrv{WTS}(\mathcal{A})$, we have that $\mathcal{M},s_0\models \varphi_1 \land \varphi_2$. To satisfy  $\varphi_1$  consider the $3$-memoryless strategy $\strat_1$ that associates $\set{\tuple{s_1,s_2}}$ to $s_1$, $\set{\tuple{s_3,s_4}}$ to $s_3$, and  $\emptyset$ to any other state of $\mathcal{M}$. Remark that for any path $\pi\in Out(s_0,\strat_1)$ and any $i\in \mathbb{N}$ we have that $\mathcal{M},\pi_i \models r_s$ iff $\pi_i\in \set{s_1,s_3,s_5}$. Thus, we must establish that $\mathcal{M}$ satisfies $\naww{\sabotage_n} \mathsf{F}(j \le 3 \ \wedge \ a)$ on $s_1$ (resp. $s_3$ and $s_5$). To do so, we remark that  $Out(s_1,\strat_1)$ (resp. $Out(s_3,\strat_1)$ and $Out(s_5,\strat_1)$) only contains the path $s_1,s_3,s_5^\omega$ (resp. $s_3,s_5^\omega$ and $s_5^\omega$) and that $\mathcal{M},s_5 \models a$. Thus, we have obtained that there is a strategy (i.e. $\strat_1$) such that for all $\pi \in Out(s_0,\strat_1)$ and all $i\in \mathbb{N}$ either $\mathcal{M},\pi_i\models \neg r_s$ or if $\mathcal{M},\pi_i \models r_s$ then there is a strategy ($\strat_1$ itself) such that $\mathcal{M},\rho_j \models a$ for some $j\geq 1$ and for all $\rho\in Out(\pi_i,\strat_1)$, as we wanted. Remark that if $t_1<3$ then it is not possible to satisfy $\varphi_1$ in $\mathcal{M}$ at $s_0$. 
For the specification $\varphi_2 = j.\naww{\sabotage_4} (\neg r_s \ \wedge \ j \le 5 \,\mathsf{W}\, a)$, consider the $4$-memoryless strategy $\strat_2$ that associates $\set{\tuple{s_0,s_1}}$ to $s_0$, $\set{\tuple{s_2,s_1},\tuple{s_2,s_3}}$ to $s_2$, $\set{\tuple{s_4,s_3}}$ to $s_4$ and $\emptyset$ to $s_5$. The only path in $Out(s_0,\strat^\star)$ is $s_0,s_2,s_4,s_5^\omega$ and since $s_5$ satisfies $a$ and all the other $s_i$ do not satisfy $ r_s$ we obtain the wanted result. 

\begin{figure}[h]
  \centering
  \includegraphics[width=77mm, height=30mm]{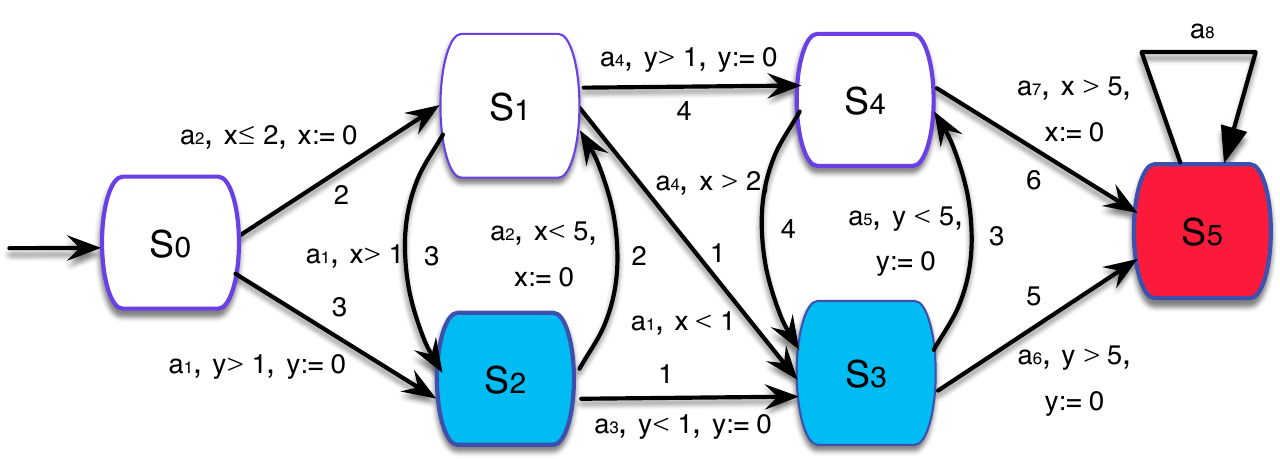}
    \vspace{-4pt}
  \caption{A \abbrv{WTA} from \cite{CLM23} where states $s_1,\, s_3$ and $s_5$ represent the goals of the attacker and the blue nodes satisfy $r_s$, the red node satisfies both $a$ and $r_s$, and the white ones satisfy neither $r_s$ nor $a$.}
  \label{fig:exp}
 \end{figure}




\section{Implementation and Validation}\label{sec:vitamin}
We implemented a \abbrv{TOL} verification algorithm within the \abbrv{VITAMIN} model checker to validate our approach and assess its runtime performance. The \abbrv{VITAMIN} tool is a prototype framework for modeling and verifying multi-agent systems, designed for both accessibility and extensibility: non-experts can use it effectively, while experts can modularly extend it with additional features~\cite{ferrando2025vitamincompositionalframeworkmodel}. To the best of our knowledge, \abbrv{VITAMIN} is currently the only tool supporting \abbrv{OL} and \abbrv{OATL}, making it a natural choice for extending with real-time verification of \abbrv{TOL}. Compared to mainstream model checkers such as \abbrv{MCMAS}~\cite{lomuscio2017mcmas} and \abbrv{STV}~\cite{stv}, implementing our algorithm in \abbrv{VITAMIN} was more straightforward and better aligned with our objectives.

A brief description of the architecture is given to familiarize readers with the tool's overall design and components, a more detailed description can be found in~\cite{ferrando2025vitamincompositionalframeworkmodel}. The main three modules depicted in Figure~\ref{fig:vitamin} are introduced below:
\begin{itemize}
    \item Logics: Holds the parsers for each supported logic, at the time of writing 12 logics are supported.
    \item Models: Responsible for storing the code that parses input files and loads them into memory as a Concurrent Game Structure (\abbrv{CGS}).
    \item Model checker/verifier: Is in control of verifying properties of models using the specified logic and model, there's one for each supported logic.
\end{itemize}

Our contribution implied changes in all modules, and enabled the introduction of a (i) parsing of real-time formulas/constraints and (ii) real-time verification algorithm for \abbrv{TCGS}s using backwards exploration of the generated zone graph, using the symbolic on-the-fly approach of Section~\ref{sec:mc}.

\begin{figure}[h]
  \centering
  \includegraphics[width=0.6\columnwidth]{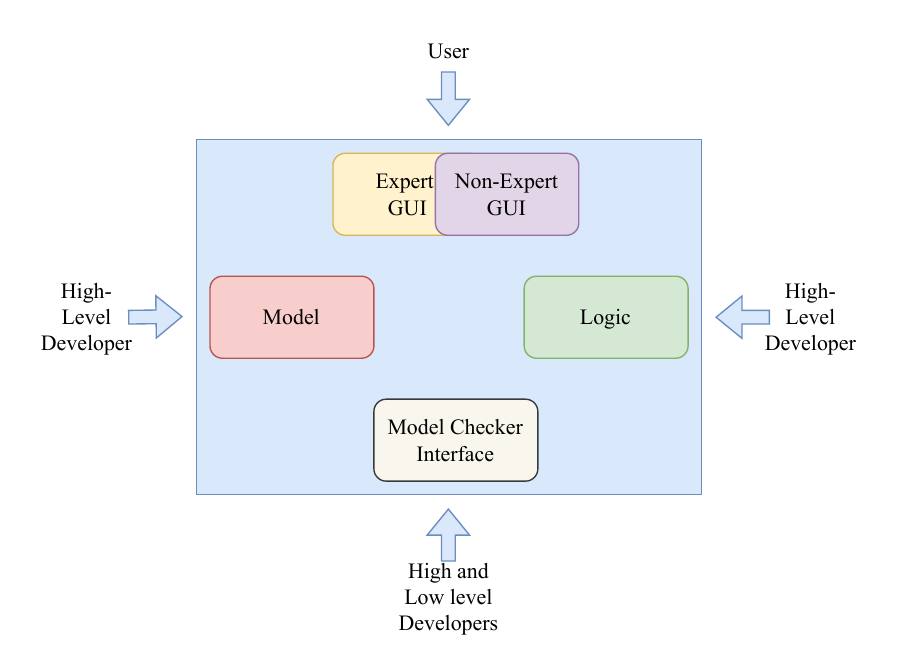}
  \caption{\abbrv{VITAMIN} high-level architecture~\cite{FerrandoMalvone2025}.}
  \label{fig:vitamin}
 \end{figure}

\subsection{Experimental methodology}
This section outlines the chosen methodology to validate the theoretical approach of the past sections. Our aim is to provide general runtime metrics of the verification algorithm with a number of selected case studies of varying size, the implementation in Python alongside the experimental data is available at: \url{https://gitlab.telecom-paris.fr/david.cortes/vitamin}.

We have crafted two models to serve as case studies, described in the following subsections. The first is the \textit{pipeline} model, a simple sequence of $k$ nodes and $k$ edges forming a single path. Finally, we examine a strongly connected model of $k$ nodes, designed to showcase \abbrv{TOL}'s ability to verify properties across paths, a capability inherited from \abbrv{CTL}. A more comprehensive experimental evaluation is planned and will be added in a subsequent version.


Although, there are publicly available benchmarks for other real-time model checkers~\cite{uppaalBenchmarks,morbe2015fully}, many of them are nearly a decade old. Additionally, many of the models or formulas are no longer available, and those that were available require substantial adaptions to be used in \abbrv{VITAMIN}. As a result, even reproducing a baseline for our experiments proved difficult, which we see as a broader limitation in the field of real-time formal verification. Furthermore, our objective is not to compare \abbrv{VITAMIN} with other more established and feature-rich tools, given that it is still a prototype direct comparisons at this stage wouldn't be meaningful.

\subsection{Case studies}
Here, we detail each model and their \abbrv{TOL} formulas used to generate the experimental data on \abbrv{VITAMIN}.
\subsubsection*{Pipeline}
The objective of this model is to test the performance of the real-time verification algorithm on simple models where there's at most two possible transitions the attacker can take at any step (to wait, or take a discrete transition). Figure~\ref{fig:pipeline-model} shows a \abbrv{WTA} for the instance where $k=4$.
\begin{figure}[h]
    \centering
    \includegraphics[width=0.5\linewidth]{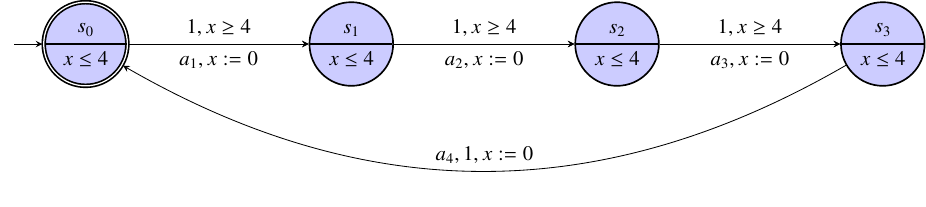}
    \caption{The pipeline model with $k=4$.}
    \label{fig:pipeline-model}
\end{figure}

The \abbrv{TOL} formula we will check against this model is $$\varphi_4 = j.\naww{\sabotage_1} \globaly(s_3 \implies j\geq 16)$$
Which captures the objective: \textit{Anytime the attacker reaches $s_3$, at least 16 time units will have passed}

And the more general formula is \textit{Anytime the attacker reaches $s_n$, at least $k*k$ time units will have passed}:
$$\varphi = j.\naww{\sabotage_1} \globaly(s_n \implies j\geq k*k)$$
The reason we chose this formula is that we can force the algorithm's worst case exploration with it, given that $s_n$ is the last location in the system, the backwards reachability algorithm starts on it, and then it will compute every time predecessor until the initial location is reached, effectively traversing most of the generated Zone Graph.

\subsubsection*{Mesh}
Figure~\ref{fig:connected-model} shows the corresponding \abbrv{WTA} for this model with $k=4$, where each node is reachable from every other node. The idea here is to collect runtime metrics that allows us to have a rough idea of how much resources could be consumed on large strongly connected graphs resembling the topology of mesh networks.
\begin{figure}[h]
    \centering
    \includegraphics[width=0.3\linewidth]{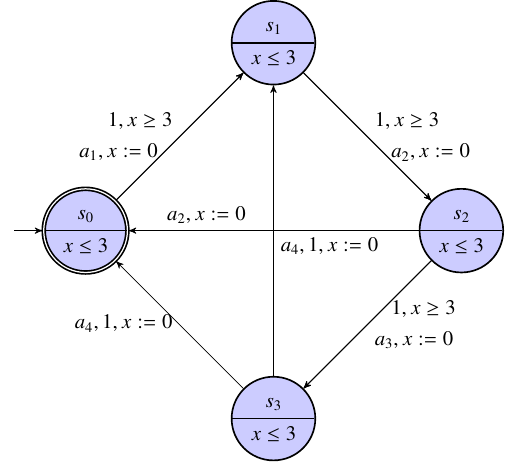}
    \caption{The strongly connected model with $k=4$.}
    \label{fig:connected-model}
\end{figure}

We're interested in running this model against the \abbrv{TOL} formula: 
$$\varphi = j.\naww{\sabotage_1} (s_n \land j\geq k*k)$$
Which means: \textit{There's a strategy for the attacker where the last location ($s_n$) is reached in at least $k*k$ time units}. There are multiple paths to $s_n$, but since we're also requiring a minimum of time units, it will traverse a sizable chunk of the Zone Graph trying to find one instance that satisfies the property. We believe this will give us reasonable performance metrics for these kinds of graphs.

\subsection{Preliminary evaluation}
Our current prototype produces results on small to medium-scale instances, which we present here as initial validation. Therefore, the following data should be regarded as preliminary, with a complete benchmark evaluation deferred to a future update.

Experiments were run on an Apple M4 system with 16~\abbrv{GB} of unified memory. For each case study, the model size was varied up to a maximum of 30.

\begin{table}
    \centering
    \begin{tabular}{|c|c|c|}
    \hline
    \textbf{Case study \& input size} & \textbf{Runtime~(ms)} & \textbf{Max memory~(\abbrv{KB})}\\
    \hline
    Pipeline $k=4$ & $7.2\pm0.8$ & $25.63\pm0.96$\\
    Pipeline $k=12$ & $13.2\pm2.3$ & $30.16\pm1.04$\\
    Pipeline $k=16$ & $19\pm6.6$ & $34.24\pm2.34$\\
    Pipeline $k=22$ & $21\pm3.4$ & $49.70\pm0.74$\\
    Pipeline $k=30$ & $25.6\pm5.3$ & $77.89\pm1.75$\\
    \hline
    \end{tabular}
    \caption{Obtained results with the pipeline model. We present the mean average runtime in milliseconds, and the mean peak allocated memory in kilobytes.}
    \label{tab:pipeline}
\end{table}

\begin{table}
    \centering
    \begin{tabular}{|c|c|c|}
    \hline
    \textbf{Case study \& input size} & \textbf{Runtime~(ms)} & \textbf{Max memory~(\abbrv{KB})}\\
    \hline
    Mesh $k=4$ & $6.2\pm1.3$ & $26.27\pm1.3$\\
    Mesh $k=12$ & $16.2\pm3.9$ &$43.42\pm2.6$\\
    Mesh $k=16$ & $19.2\pm4.1$ & $67.11\pm3.9$\\
    Mesh $k=22$ & $24.4\pm2.4$ & $112.86\pm2.5$\\
    Mesh $k=30$ & $37.8\pm3.4$ & $196.65\pm2.1$\\
    \hline
    \end{tabular}
    \caption{Obtained results with the mesh model. We present the mean average runtime in milliseconds, and the mean peak allocated memory in kilobytes.}
    \label{tab:strongly-connected}
\end{table}

We measured the average runtime and max allocated memory over five runs for every case study and input size, using Python's tracemalloc and time modules. Then, the collected data were compiled into multiple CSV files and processed using Python notebooks. The aggregated results are shown in tables~\ref{tab:pipeline} and~\ref{tab:strongly-connected}. Expectedly, we see the metrics increasing steadily with input size in both tables. The mesh model generally consumed more memory than the pipeline model, while keeping comparable execution times, except for the largest input size. The Python notebook along with the raw collected data is available at \url{https://gitlab.telecom-paris.fr/david.cortes/vitamin-benchmarks}.

\subsection{Challenges and future work}
Since \abbrv{VITAMIN} is still in prototype phase, certain engineering aspects remain under development. This led to challenges such as build issues related to missing dependencies or unsupported versions, opening an avenue for exploring containerized solutions to streamline time to development. Additionally, automating the generation and collection of experimental data proved difficult given there's no command-line access to the tool, we plan on working on this limitation next to support and provide much more comprehensive experimental analysis for an upcoming version of this paper. 

We emphasize that these challenges are not fundamental obstacles but reflect the current prototype status. Addressing them will further strengthen the tool as it evolves into a stable and extensible platform.

\section{Related Work}\label{sec:relwork}

Lately, many papers have focused on the strategic capabilities of agents playing within dynamic game models. In this section, we compare our approach with some of these papers. 

\textbf{Untimed Games and Strategic Logics}~\cite{vanBenthem2005,Lding2003ModelCA,aucherSML} some research related to sabotage games  have been introduced by van Benthem with the aim of studying the computational complexity of a special class of graph reachability problems in which an agent has the ability to delete edges. 
%
To reason about sabotage games, van Benthem introduced Sabotage Modal Logic (\abbrv{SML}). The model checking problem for the sabotage modal logic is \abbrv{PSPACE-complete}~\cite{Lding2003ModelCA}. Our version of the games is not comparable to the sabotage games, because we provide the possibility to temporarily select subsets of edges, while in the sabotage games the saboteur can only delete one edge at a time. In this respect, our work is related to~\cite{CattaLM23}, where the authors use an extended version of sabotage modal logic, called Subset Sabotage Modal Logic (\abbrv{SSML}), which allows for the deactivation of certain subsets of edges of a directed graph. The authors show that the model checking problems for such logics are decidable. 
Furthermore, we recall that \abbrv{SSML} is an extension of \abbrv{SML}, but does not include temporal operators. Also, neither \abbrv{SML} nor \abbrv{SSML} takes into account quantitative information about the cost of edges, as we do.
In \cite{DynamicVadim} Dynamic Escape Games (\abbrv{DEG}) have been introduced.  A \abbrv{DEG} is a variant of weighted two-player turn-based reachability games. In a \abbrv{DEG}, an agent has the ability to inhibit edges. In \cite{CLM23} have been introduced an untimed Obstruction Logic (\abbrv{OL}) which allows reasoning about two-player games played on a labeled and weighted directed graph. Thus, our \abbrv{TOL} is an extension of \abbrv{OL}. \abbrv{ATL} \cite{AHK02} and \abbrv{SCTL} \cite{APS23} are extensions of \abbrv{CTL} with the notion of strategic modality. These kinds of logics are used to express properties of agents as their possible actions. However, all these logics do not include quantitative information about costing edges, real-time and temporal operators.

%


\textbf{Timed Games and Strategic Logics}~\cite{HPV06, KnapikPetrucciJamrogaTimed, FTM14, QAS19, APS23} several research works have focused on extending games and modal and temporal logics to the real-time domain.  The most established model in this respect is the timed game automata \cite{BFLM11, FTM14}. A Timed Game Automaton (\abbrv{TGA}) is a \abbrv{TA} whose set of transitions is divided among the different players. At each step, each player chooses one of her possible transitions, as well as some time she wants to wait before firing her chosen transition.  The logics \abbrv{ATL} and \abbrv{CTL} has also been extended to \abbrv{TATL} \cite{HPV06, KnapikPetrucciJamrogaTimed} and \abbrv{STCTL} \cite{APS23}, in which formula clocks are used to express the time requirements. It is exponentially decidable whether a \abbrv{TATL} formula satisfies a \abbrv{TGA} \cite{HPV06}. However, in \cite{APS23} was shown that \abbrv{STCTL} is more expressive than \abbrv{TATL} and the model checking problem for \abbrv{STCTL}  with  memoryless perfect information is of the same complexity as for \abbrv{TCTL}. Model checking for \abbrv{TATL} with continuous semantics is undecidable \cite{APS23}.  However, all these logics  do not use dynamic models.

\section{Conclusions}\label{sec:end}
In this paper, we presented Timed Obstruction Logic (\abbrv{TOL}), a logic that allows reasoning about two-player games with real-time goals, where one of the players has the power to locally and temporarily modify the timed game structure. We then proved that its model checking problem is \abbrv{PSPACE-complete}. Furthermore, we showed how \abbrv{TOL} expresses cybersecurity properties in a suitable way. There are several directions that we would like to explore for future work. 
A possible extension would be to consider timed games with many players, between a demon and \emph{coalitions} of travelers. 
In our opinion, such an extension would have the same relationship with the \abbrv{TATL} logic as \abbrv{TOL} has with \abbrv{TCTL}. Another extension could be to extend TOL with probability, since many cybersecurity case studies relate cyberattacks to the probability of success of the attack. Finally, another important line would be to introduce imperfect information in our setting. Unfortunately, this context is generally non-decidable \cite{DimaT11}. In order to overcome this problem, we could use an approximation to perfect information \cite{BelardinelliFM23}, a notion of bounded memory \cite{BelardinelliLMY22}, or some hybrid technique \cite{FerrandoM22,FerrandoM23}.

\bibliography{biblio}
\end{document}